\documentclass[journal]{IEEEtran}
\usepackage{color}
\usepackage{amsmath}
\usepackage{amssymb}
\usepackage{dsfont}
\usepackage{upgreek}
\usepackage{booktabs}
\usepackage{multirow}
\usepackage[table,xcdraw]{xcolor}
\usepackage{subfigure}
\usepackage{bm}
\usepackage{xspace}
\usepackage{hyperref}

\usepackage{ifpdf}
 \ifpdf
 \else
 \fi

\usepackage{cite}

\ifCLASSINFOpdf
  \usepackage[pdftex]{graphicx}
  \graphicspath{{../pdf/}{../jpeg/}}
  \DeclareGraphicsExtensions{.pdf,.jpeg,.png}
\else
  \usepackage[dvips]{graphicx}
  \graphicspath{{../eps/}}
  \DeclareGraphicsExtensions{.eps}
\fi
\usepackage{caption2}  

\usepackage{amsmath}
\usepackage{ntheorem}
\theoremstyle{plain}
\newtheorem{theorem}{Theorem}

\newtheorem*{proof}{Proof}
\theoremstyle{definition}
\newtheorem{definition}{Definition}
\newtheorem{example}{Example}
\usepackage[numbers,sort&compress]{natbib}
\usepackage{enumerate}
\usepackage{algorithm}
\usepackage{algorithmic}

\usepackage{array}

\usepackage{url}

\usepackage{multirow}
\usepackage{bm}

\begin{document}
%
\title{Maximal co-occurrence nonoverlapping sequential rule mining}
%
%
%

\author{Yan Li, Chang Zhang, Jie Li, Wei Song, Zhenlian Qi, Youxi Wu,~\IEEEmembership{Member,~IEEE,} and Xindong Wu,~\IEEEmembership{Fellow,~IEEE}
\thanks{Manuscript received August 22, 2022.  (Corresponding author: Y. Wu). }
\thanks{Yan Li, Chang Zhang, Jie Li are with the School of Economics and Management, Hebei University of Technology, Tianjin, 300400, China (e-mail: lywuc@163.com.)}
\thanks {Wei Song is with the North China University of Technology, Beijing, China, (e-mail: songwei@ncut.edu.cn)}
\thanks {Zhenlian Qi is with Guangdong Eco-Engineering Polytechnic, Guangzhou 510520, China, (e-mail: qzlhit@foxmail.com)}
\thanks{Youxi Wu is with the School of Artificial Intelligence, Hebei University of Technology, Tianjin, 300400, China (e-mail: wuc567@163.com.)}
\thanks{Xindong Wu is with Key Laboratory of Knowledge Engineering with Big Data (the Ministry of Education of China), Hefei University of Technology, Hefei, 230009, China (e-mail: xwu@hfut.edu.cn)}
}

%
%

\markboth{IEEE Transactions on Knowledge and Data Engineering }
{Shell \MakeLowercase{\textit{Ma et al.}}: Bare Demo of IEEEtran.cls for IEEE Journals}
%



\maketitle

\begin{abstract}
The aim of sequential pattern mining (SPM) is to discover potentially useful information from a given sequence. Although various SPM methods have been investigated, most of these focus on mining all of the patterns. However, users sometimes want to mine patterns with the same specific prefix pattern, called co-occurrence pattern. Since sequential rule mining can make better use of the results of SPM, and obtain better recommendation performance, this paper addresses the issue of maximal co-occurrence nonoverlapping sequential rule (MCoR) mining and proposes the MCoR-Miner algorithm. To improve the efficiency of support calculation, MCoR-Miner employs depth-first search and backtracking strategies equipped with an indexing mechanism to avoid the use of sequential searching. To obviate useless support calculations for some sequences, MCoR-Miner adopts a filtering strategy to prune the sequences without the prefix pattern. To reduce the number of candidate patterns, MCoR-Miner applies the frequent item and binomial enumeration tree strategies. To avoid searching for the maximal rules through brute force, MCoR-Miner uses a screening strategy. To validate the performance of MCoR-Miner, eleven competitive algorithms were conducted on eight sequences. Our experimental results showed that MCoR-Miner outperformed other competitive algorithms, and yielded better recommendation performance than frequent co-occurrence pattern mining. All algorithms and datasets can be downloaded from https://github.com/wuc567/Pattern-Mining/tree/master/MCoR-Miner. 

\end{abstract}


%
\IEEEpeerreviewmaketitle

\section{Introduction}
%

\IEEEPARstart{A}{s} an important method of knowledge discovery \cite {wuxindong2022tmis}, sequential pattern mining (SPM) \cite {spmacm} aims to mine sub-sequences (patterns) that meet certain conditions from sequence datasets \cite {gan2019tkdd}. A variety of SPM methods have been derived for different mining requirements, such as order-preserving SPM for time series \cite {wu2022orde}, SPM for large-scale databases \cite {linlarge2021, itjlarge},  episode pattern  mining \cite {episodepattern}, spatial co-location pattern mining \cite {wanglizhen2018, wanglizhen2022}, contrast SPM \cite {wu2021constrast, jessica2020}, negative SPM \cite { Dong2019tnnls, wu2022negative}, high utility SPM \cite {gan2021tkde}, high average-utility SPM \cite {Truongtkde,weisong2021, jinkais2020}, outlying SPM \cite {duanlei2020tkdd}, three-way SPM \cite {Min2020ins, wu2021ntpm}, co-occurrence SPM \cite {cooccpattern, generator2020kais}, and SPM with gap constraints \cite {wang2022apin}. One of the disadvantages of traditional SPM is that it only considers whether a given pattern occurs within a sequence and ignores the repetition of the pattern in the sequence \cite {wu2018tcyb}. For example, the support (number of occurrences) of a pattern $ad$ in sequence $adbdad$ is one according to traditional SPM, despite the pattern $ad$ occurring more than once in the sequence. From this example, we see that repetition is ignored in traditional SPM, meaning that some interesting patterns will be lost \cite {gapspm}.

To solve this issue, gap constraint SPM was proposed  \cite {zhang2007tkdd}, which can be expressed as \textbf{p} = $p_1p_2$ $\cdots$ $p_{m-1}p_m$ with $gap$ = $[a,b]$ (or \textbf{p} = $p_1[a,b]p_2$…$[a,b]p_m$), where $a$ and $b$ $(0 \leq a \leq b)$ are integers indicating the minimum and maximum wildcards between $p_j$ and $p_{j+1}$, respectively \cite {shi2020apin, li2012tkdd}. For example,  pattern $a[0,3]d$ means that there are zero to three wildcards between $a$ and $d$. Repetitions of patterns can be useful in terms of capturing valuable information from sequences, and gap constraint SPM therefore has many applications, such as septic shock prediction for ICU patients \cite {lijinyan2017}, keyphrase extraction \cite{keywu2017}, missing behaviors analysis \cite{Cao2021}, and pyramid scheme pattern mining \cite{pyramidpattern}. However, gap constraint SPM is not only difficult to solve, but also has many forms, such as periodic gaps SPM \cite {wu2014apin}, disjoint SPM \cite {disjoint}, one-off SPM \cite {wu2021tmis}, and nonoverlapping SPM \cite {wang2022apin}. Previous research work has shown that the nonoverlapping SPM avoids the production of many redundant patterns \cite {wu2018tcyb}; for example, under the nonoverlapping condition, the pattern $ad$ occurs twice in sequence $adbdad$. 

Unfortunately, current schemes based on gap constraint SPM can only mine all of the frequent patterns. Compared with rule mining, frequent pattern mining is not suitable for making predictions or recommendations. Rule mining \cite {fournierviger2015tkde} focuses on mining rules such as \textbf{p} $\to$ \textbf{q}, which means that if pattern \textbf{p} occurs in the sequence, pattern \textbf{q} is likely to appear afterward with a probability higher than or equal to a given confidence, i.e., $conf($\textbf{p} $\to$ \textbf{q}$)$ $\geq$ $mincf$. However, in cases where users only want to find the rules with the same antecedent \textbf{p}, then mining all the strong rules is not only time-consuming and laborious, but also meaningless. For example, in a recommendation problem, researchers hope to discover potential rules from historical data and use the recent events to predict future events. Obviously, it is valuable to discover rules with the same recent events instead of all rules. This kind of rule is called co-occurrence rule and is more meaningful. For clarification, an illustrative example is as follows.

\begin{example}\label{example1}

Suppose we have a sequence \textbf{s} = $adbdadcdccabadcd$, a gap constraint $gap$ = [0,3], and a predefined support threshold \textit{minsup} = 3. According to nonoverlapping SPM, there are 12 frequent patterns: \{$a$, $c$, $d$, $ad$, $cd$, $dc$, $dd$, $adc$, $add$, $dcd$, $ddc$, $adcd$\}. Moreover, if we continue to set the minimum confidence threshold \textit{mincf} = 0.7, there are 11 rules: \{$a$ $\to$ $d$, $a$ $\to$ $dc$, $a$ $\to$ $dd$, $a$ $\to$ $dcd$, $d$ $\to$ $d$, $d$ $\to$ $c$, $ad$ $\to$ $d$, $ad$ $\to$ $c$, $ad$ $\to$ $cd$, $dd$ $\to$ $c$, $adc$ $\to$ $d$\}. Obviously, it is difficult for users to apply these excessive numbers of frequent patterns and rules.

However, if we have a prefix-pattern (or antecedent) \textbf{p} = $ad$, the number of frequent super-patterns and rules will be reduced. The details are shown as follows. We know that \textbf{p} = $ad$ occurs four times in sequence \textbf{s}, since the subsequence $s_1s_2$ is $ad$ and the subsequences $s_5s_6$, $s_{11}s_{14}$, and $s_{13}s_{16}$ are also $ad$. Thus, the support of \textbf{p} in \textbf{s} is four. Similarly, we know that the subsequences $s_1s_4s_7$, $s_5s_6s_9$, and $s_{11}s_{14}s_{15}$ are all $adc$. Thus, the support of pattern $adc$ in \textbf{s} is three, which is not less than $minsup$ = 3. Hence, the pattern $adc$ is a frequent co-occurrence pattern of \textbf{p} = $ad$. Furthermore, $ad$ $\to$ $c$ is a co-occurrence rule whose confidence is 3/4 = 0.75 $\geq$ $mincf$ = 0.7. Similarly, we know that patterns $add$ and $adcd$ are frequent, and that $ad$ $\to$ $d$ and $ad$ $\to$ $cd$ are co-occurrence rules. Hence, there are only three co-occurrence patterns of pattern \textbf{p}, $adc$, $add$, and $adcd$, and three co-occurrence rules $ad$ $\to$ $c$, $ad$ $\to$ $d$, and $ad$ $\to$ $cd$, which means that the number of frequent patterns and rules is greatly reduced.
\end{example}

Inspired by the maximal SPM, the concept of maximal co-occurrence rules (MCoRs) is developed to further decrease the number of co-occurrence rules. For instance, if a rule \textbf{p} $\to$ \textbf{r} is an MCoR, then we know that \textbf{p} $\to$ \textbf{q} is also a co-occurrence rule, where pattern \textbf{q} is the prefix pattern of \textbf{r}. Moreover, if the pattern \textbf{w} is a superpattern of \textbf{r}, then \textbf{p} $\to$ \textbf{w} is not a co-occurrence rule. From the above examples, we see that it is meaningful to investigate MCoR mining. The main contributions of this paper are as follows. 

\begin{enumerate}[1.]
\item  To avoid mining irrelevant patterns and obtain better recommendation performance, we develop MCoR mining, which can mine all MCoRs with the same rule-antecedent, and we propose the MCoR-Miner algorithm. 

\item In MCoR mining, the user does not need to set the support threshold, since it can be automatically calculated based on the support of the rule-antecedent and the minimum confidence threshold.

\item To improve support calculation, MCoR-Miner consists of three parts: a preparation stage, candidate pattern generation, and a screening strategy to improve efficiency.


\item To validate the performance of MCoR-Miner, eleven competitive algorithms and eight datasets are selected. Our experimental results verify that MCoR-Miner outperforms the other competitive algorithms and yields better recommendation performance than frequent SPM.
\end{enumerate}

The structure of this paper is as follows. Section \ref{section:Related work} gives an overview of related work. Section \ref{section:Problem definitions} defines the problem. Section \ref{section:Proposed algorithm} presents the MCoR-Miner algorithm. Section \ref{section:Experimental results and analysis} reports the performance of MCoR-Miner. Section \ref{section:CONCLUSION} concludes this paper.

\section{Related work}
\label{section:Related work}

The aim of SPM is to find subsequences (patterns) in a sequence database that meet the given requirements \cite {Okolica2020tkde}. SPM methods are commonly used for data mining, since their results are intuitive and interpretable \cite {smedt2020tkde,forest2022tkdd}. Traditional SPM mainly focuses on frequent pattern mining, which means the mined patterns have high frequency. To reduce the number of patterns, top-$k$ SPM \cite {Dong2019tnnls}, closed SPM \cite {closed2020kbs}, and maximal SPM \cite {li2022apinm} were developed, all of which require that the data are static rather than dynamic. To overcome this drawback, incremental SPM \cite {incremental} and window SPM \cite {window} methods were explored. However, the research community noticed that rare patterns were of great significance in the field, and rare pattern mining was also proposed \cite {rare2016tkdd}. In addition, to discover missing events, negative SPM methods such as e-NSP \cite {Cao2016ensp} and e-RNSP \cite {Dong2020ernsp} were designed. An approach called high utility SPM was also designed \cite {jlintkdd2022, gantcyb}, which can discover patterns with low frequency but high utility \cite {wulei2021eswa, Nawaz2022tmis}.

Most SPM methods can be seen as classical SPM \cite {gantkdd}, since these methods only consider whether a pattern occurs within a sequence, while repetitive SPM deeply considers the number of occurrences of a pattern in the sequence \cite {wu2014apin}. For example, suppose we have two sequences: \textbf{s}${_1}$=aabbaaba and \textbf{s}${_2}$=abcbc. According to classical SPM, the support of pattern aba in \textbf{s}${_1}$ and \textbf{s}${_2}$ is one, since pattern aba occurs in sequence \textbf{s}${_1}$=aabbaaba, but does not occur in sequence \textbf{s}${_2}$. Thus, classical SPM neglects the fact that pattern aba occurs in sequence \textbf{s}${_1}$=aabbaaba more than once, while repetitive SPM considers the number of occurrences. Another significant difference is that classical SPM focuses on mining  patterns in sequences with itemsets, while repetitive SPM mainly aims to mine  patterns in sequences with items \cite{pmdb2022six}. For example, (ab)(bc)(a)(bd)(ad) is a sequence with itemsets, and each itemset has many ordered items. If all itemsets in a sequence have only one item, then the sequence is a sequence with items. Thus, a sequence with items can be seen as a special case of a sequence with itemsets. The sequences with items are used in many fields, such as DNA sequence, protein sequence, clickstreams, and commercial data. 

Note that repetitive SPM is similar to episode mining \cite {episodemaximal}. The differences are three-fold. First, episode mining deals with one sequence, while repetitive SPM processes one or more sequences. Second, episode mining aims to mine patterns in an event sequence, where an event sequence can be represented by $<$($e_1$,$t_1$), ($e_2$,$t_2$), $\cdot$ ($e_n$,$t_n$)$>$, where $e_i$ is an event set, and $t_i$ is the occurrence time of $e_i$. In contrast, in repetitive SPM, the sequences do not have the occurrence time. Third, in episode mining, $e_i$ can be a set, while in repetitive SPM, the sequence consists of items, rather than sets.

Compared with classical SPM, repetitive SPM not only is more challenging, but  also has many forms: general form (no condition) \cite {wu2014apin}, disjoint  \cite{pmdb2022six}, one-off \cite {li2022apind}, and nonoverlapping forms \cite {wu2021insweak}. Note that the disjoint form was called the nonoverlapping form in some studies \cite {Dong2020ernsp, disjoint}, which is far different from the nonoverlapping form in this study. To clarify the difference between disjoint and nonoverlapping, an illustrative example is shown as follows.

\begin{example}\label{exampledisjoint}
Suppose we have a sequence \textbf{s}=aabbaaba and a pattern \textbf{p}=a[0,1]b[0,1]a.

In the disjoint form \cite {Dong2020ernsp, disjoint}, the first position of an occurrence is greater than the last position of its previous occurrence. Thus, there are two occurrences: $<$1,3,5$>$ and $<$6,7,8$>$. Note that $<$1,3,5$>$ and $<$2,4,6$>$ do not satisfy the disjoint form, since the first position of occurrence $<$2,4,6$>$ 2 is less than 5, which is the last position of $<$1,3,5$>$.

In the nonoverlapping form \cite {wu2018tcyb}, each item cannot be reused by the same $p_j$, but can be reused by different $p_j$. Thus, there are three occurrences: $<$1,3,5$>$,$<$2,4,6$>$, and $<$6,7,8$>$. Note that $<$2,4,6$>$, and $<$6,7,8$>$ satisfy the nonoverlapping form, since  $p_3$ matches $s_6$ in $<$2,4,6$>$, and $p_1$ matches $s_6$ in $<$6,7,8$>$.
\end{example}

Recently, various applications for SPM with gap constraints were investigated. For example, top-$k$ contrast nonoverlapping SPM was proposed, in which the mined patterns can be used as features for a sequence classification task \cite {wu2021constrast}. To discover  low frequency but high average utility patterns, high average utility one-off SPM \cite {wulei2021eswa} and nonoverlapping SPM \cite {dfom2021kbs} were developed. To discover missing events, a one-off negative SPM was proposed, which could be used to predict the future trend in traffic flow \cite {wu2022negative}. Inspired by three-way decisions \cite {Min2020ins, zhanthreeway}, nonoverlapping three-way SPM \cite {wu2021ntpm} was designed to mine the patterns to which users pay the most attention, and can effectively mine patterns composed of strong-interest and medium-interest items.

Most of these schemes aim to discover all of the patterns that satisfy the predefined constraints. However, in general, these methods will discover numerous patterns. Although top-$k$ SPM \cite {wu2021constrast}, nonoverlapping closed SPM  \cite {closed2020kbs}, and nonoverlapping maximal SPM \cite {li2022apinm} can reduce the number of patterns, users may not be interested in many of the patterns that are discovered. However, these mining methods fail to reveal the relationship between patterns. Rule mining is an effective method in discovering the relationship between patterns \cite{fournierviger2015tkde}. There are many rule mining methods such as association rule mining \cite {assrule}, sequential rule mining \cite {fournierviger2015tkde}, episode rule mining \cite {episoderule}, and order-persevering rule mining \cite{wu2022TKDE}. Similar to pattern mining, there are many methods to reduce the redundant rules, such as maximum consequent and minimum antecedent \cite {rulemining}, top-k rule mining \cite {assrule}, closed rule mining, and maximal rule mining \cite{wu2022TKDE}. Nevertheless, in some cases, users know a prefix pattern advance, and they want to discover its super-patterns. This is called co-occurrence pattern mining \cite {co2019apin}.  Based on the co-occurrence patterns, co-occurrence rules can be further explored to reveal the relationships between the prefix patterns and their super-patterns. Although we can mine all patterns at first, and then filter out useless patterns with different prefixes, this approach will increase the running time \cite {targetpattern}.

In summary, although the use of SPM with gap constraints can make the mining results more meaningful, the mining results of this scheme are not targeted. In order to make the mining results more suitable for recommendations, based on the special needs of users, we present a scheme inspired by co-occurrence SPM \cite {coocctkde} and maximal nonoverlapping SPM \cite {li2022apinm}, called MCoR mining.

\section{Problem definitions}
\label{section:Problem definitions}

\begin{definition}\label{definition1}
  \rm (Sequence) A sequence \textbf{s} with length $n$ is denoted by \textbf{s} = $s_1s_2$ $\cdots$ $s_n$, where $s_i$ (1 $\leq i \leq n$) $\in \sum$, $\sum$ represents a set of items in sequence \textbf{s}, and the size of $\sum$ can be expressed as $|\sum|$. 
\end{definition}

\begin{definition}\label{definition2}
   \rm (Sequence database) A sequence database $D$ with length $k$ is a set of sequences, denoted by $D$ = \{$s_1$, $s_2$, $\cdots,$ $s_k$\}.
\end{definition}

\begin{definition}\label{definition3}
  \rm  (Pattern) A pattern \textbf{p} with length $m$ is denoted by \textbf{p} = $p_1[a,b]$ $p_2$ $\cdots$ $[a,b]p_m$ (or abbreviated as \textbf{p} = $p_1$$p_2$ $\cdots$ $p_{m-1}$$p_m$ with $gap$ = $[a,b])$, where $a$ and $b$ $(0\leq a \leq b)$ are integers indicating the minimum and maximum wildcards between $p_j$ and $p_{j+1}$, respectively.
\end{definition}

\begin{definition}\label{definition4}
   \rm (Occurrence and nonoverlapping occurrence) Suppose we have a sequence \textbf{s} = $s_1$ $s_2$ $\cdots$ $s_n$ and a pattern \textbf{p} = $p_1[a,b]$$p_2 $ $\cdots $$[a,b]p_m$. $l$ = $<$$l_1$, $l_2$, $\cdots$, $l_m$$>$ is an occurrence of pattern \textbf{p} in sequence \textbf{s} if and only if $p_1$ = $s_{l_1}$, $p_2$ = $s_{l_2}$, $\cdots$, $p_m$ = $s_{l_m}$ (0 $\leq l_1$ $\leq l_2$ $\leq \cdots$ $\leq l_m$ $\leq n$) and $a \leq$ $l_i-l_{i-1}-1$ $\leq b$. We assume there is also another occurrence $l^{'}$ = $<$$l_1^{'}$, $l_2^{'}$, $\cdots$, $l_m^{'}$$>$. $l$ and $l^{'}$ are two nonoverlapping occurrences if and only if for any $(1 \leq j \leq m)$, $l_j \neq  l_j^{'}$.
\end{definition}

\begin{example}\label{example3}
 Suppose we have a sequence \textbf{s} = $adbdadcdccabadcd$ and a pattern \textbf{p} = $a[0,3]d$. According to the gap constraint [0,3], all occurrences of \textbf{p} in \textbf{s} are $<$1,2$>$, $<$1,4$>$, $<$5,6$>$, $<$5,8$>$, $<$11,14$>$, $<$13,14$>$ and $<$13,16$>$. $<$1,2$>$ and $<$1,4$>$ do not satisfy the nonoverlapping condition, since 1 appears in these two occurrences in the same position. Thus, there are four nonoverlapping occurrences of \textbf{p} in \textbf{s}, which are $<$1,2$>$, $<$5,6$>$, $<$11,14$>$, and $<$13,16$>$. 
 
 Note that there are many different methods to calculate the nonoverlapping occurrences and the results may be different, such as NETLAP \cite{netlap}, NETGAP \cite{wu2018tcyb}, Netback \cite {li2022apinm}, DFOM \cite {dfom2021kbs}. For example, besides two nonoverlapping occurrences $<$1,2$>$ and $<$5,6$>$, $<$1,4$>$ and $<$5,8$>$ are also two nonoverlapping occurrences. Although the results may be different, this problem has been theoretically proved to be solved in polynomial time \cite{netlap}. More importantly, it has been shown that there are four different ways to find the nonoverlapping occurrences, finding the maximal occurrences in the rightmost leaf-root way, finding the maximal occurrences in the rightmost root-leaf way, finding the minimal occurrences in the leftmost leaf-root way, and finding the minimal occurrences in the leftmost root-leaf way \cite{jos2021}. In this example, $<$1,2$>$ and $<$5,6$>$ are called minimal nonoverlapping occurrences \cite{dfom2021kbs}, while $<$1,4$>$ and $<$5,8$>$ are called maximal nonoverlapping occurrences \cite{netlap}.  In this paper, we search for the minimal nonoverlapping occurrences.
\end{example}

\begin{definition}\label{definition5}
  \rm  (Support) The support of pattern \textbf{p} in sequence \textbf{s} is the number of nonoverlapping occurrences, represented by $sup(\textbf{p},\textbf{s})$. The support of pattern \textbf{p} in sequence database $D$ is the sum of the supports in each sequence, i.e., $sup(\textbf{p},D) $ = $\sum_{l=1}^{k}sup(\textbf{p},\textbf{s}_i)$.
\end{definition}

\begin{definition}\label{definition6}
  \rm (Frequent pattern) If the support of pattern \textbf{p} in sequence \textbf{s} or $D$ is no less than the minimum support threshold $minsup$, then pattern \textbf{p} is a frequent pattern.
\end{definition}

\begin{definition}\label{definition7}
  \rm  (Prefix pattern, subpattern, and superpattern). Given two patterns \textbf{p} = $p_1p_2$ $\cdots$ $p_{m-1}p_m$ and \textbf{q} = $q_1q_2$ $\cdots$ $q_{n-1}q_n$ ($m < n$) with $gap$ = [$a,b$], if and only if $p_1$ = $q_1$, $p_2$ = $q_2$, $\cdots$, and $p_m$ = $q_m$, then pattern \textbf{p} is the prefix pattern of pattern \textbf{q}. Moreover, pattern \textbf{p} is a subpattern of pattern \textbf{q}, and pattern \textbf{q} is a superpattern of pattern \textbf{p}.
\end{definition}

\begin{example}\label{example4}
    Given two patterns \textbf{p} = $ad$ and \textbf{q} = $adbd$, pattern \textbf{p} is the prefix pattern of pattern \textbf{q}.
\end{example}

\begin{definition}\label{definition8}
   \rm (Co-occurrence pattern, co-occurrence rule, rule antecedent, rule consequent, and confidence) Suppose we have a pattern \textbf{q}=\textbf{p}·\textbf{r} = $p_1p_2$ $\cdots$ $p_m q_{m+1}q_{m+2}$ $\cdots$ $q_n$, where \textbf{p} = $p_1p_2$ $\cdots$ $p_m$ and \textbf{r} = $q_{m+1}q_{m+2}$ $\cdots$ $q_n$.   The pattern \textbf{r} is a co-occurrence pattern of \textbf{p}, and \textbf{p} $\to$ \textbf{r} is a co-occurrence rule, where \textbf{p} and \textbf{r} are the rule antecedent and consequent, respectively. The ratio of the supports of patterns \textbf{q} and \textbf{p} is called the confidence of the sequential rule \textbf{p} $\to$ \textbf{r}, and is denoted by \textit{conf}(\textbf{p} $\to$ \textbf{r}) = sup(\textbf{p}·\textbf{r},$D$)/sup(\textbf{p},$D$). The confidence indicates the probability of occurrence of pattern \textbf{r} when pattern \textbf{p} occurs.
\end{definition}

\begin {definition}
\rm(Maximal co-occurrence pattern) Suppose we have a pattern \textbf{q}. If one of its superpattern \textbf{r} is a frequent co-occurrence pattern, then pattern \textbf{q} is not a maximal co-occurrence pattern; otherwise, pattern \textbf{q} is a maximal co-occurrence pattern. 
\end {definition}

\begin{definition}\label{definition9}
  \rm (Strong co-occurrence rule and MCoR) If $conf($\textbf{p} $\to$ \textbf{r}$)$ is greater than or equal to the predefined threshold $mincf$, i.e., $conf($\textbf{p} $\to$ \textbf{r}$)$ $\geq$ $mincf$, then \textbf{p} $\to$ \textbf{r} is a strong co-occurrence rule. Suppose \textbf{p} $\to$ \textbf{r} is a strong co-occurrence rule.  For any superpattern \textbf{w} of pattern \textbf{r}, if \textbf{p} $\to$ \textbf{w} is not a strong co-occurrence rule, then \textbf{p} $\to$ \textbf{r} is an MCoR. 
\end{definition}

\textbf{MCoR mining:} Given a sequence \textbf{s} or sequence dataset $D$, $mincf$, and prefix pattern \textbf{p}, the aim of MCoR mining is to discover all MCoRs. 

\begin{example}\label{example5}
  From Example \ref{example1}, we know that $ad$ $\to$ $c$ and $ad$ $\to$ $cd$ are two co-occurrence rules. Rule $ad$ $\to$ $c$ is not an MCoR, while rule $ad$ $\to$ $cd$ is an MCoR, since the pattern $c$ is the prefix pattern of $cd$ and for any superpattern of $cd$, such as $cdc$, $ad$ $\to$ $cdc$ is not a strong co-occurrence rule. 
\end{example}

The symbols used in this paper are shown in Table \ref{Notation}.

\begin{table}
\centering
    \scriptsize
    \caption{Notations}
       \begin{tabular}{lc}
        \toprule
        {Symbol}     & {Description}  \\
        \midrule
        {\textbf{s}}  & A sequence  with length $n$  \\
        {$D$}  & A sequence database with $k$ sequences  \\
        {\textbf{p}}  & A pattern with length $m$  \\
        {$a,b$}  & The minimum and maximum wildcards, respectively   \\
         {$sup(\textbf{p},\textbf{s})$} & The number of  occurrences of \textbf{p} in \textbf{s}  \\
         {$sup(\textbf{p},D) $}  & The number of  occurrences of \textbf{p} in $D$   \\
         {$\sum $}  & The set of items in sequence database $D$    \\
         {$minsup$}  & The minimum support threshold    \\
         {$mincf$}  & The predefined confidence threshold    \\
          {\textbf{p} $\to$ \textbf{r}}  & A co-occurrence  rule  \\
          {\textit{conf}(\textbf{p} $\to$ \textbf{r})}  & The confidence of the co-occurrence  rule \textbf{p} $\to$ \textbf{r} \\
         
        \bottomrule
    \end{tabular}
	\label{Notation}
\end{table}

\section{Proposed algorithm}
\label{section:Proposed algorithm}

In this section, we propose MCoR-Miner, an algorithm designed to discover all MCoRs. MCoR mining is based on nonoverlapping SPM whose main issue is support calculation. Therefore, support calculation is a key aspect of MCoR mining illustrated in Section \ref {subsection:Support calculation}. In addition to support calculation, MCoR-Miner has three parts: preparation stage, candidate pattern generation, and screening strategy.  The framework of MCoR-Miner is shown in Fig. \ref{Framework}. 



\begin{figure}
    \centering
    \includegraphics[width=0.95\linewidth]{"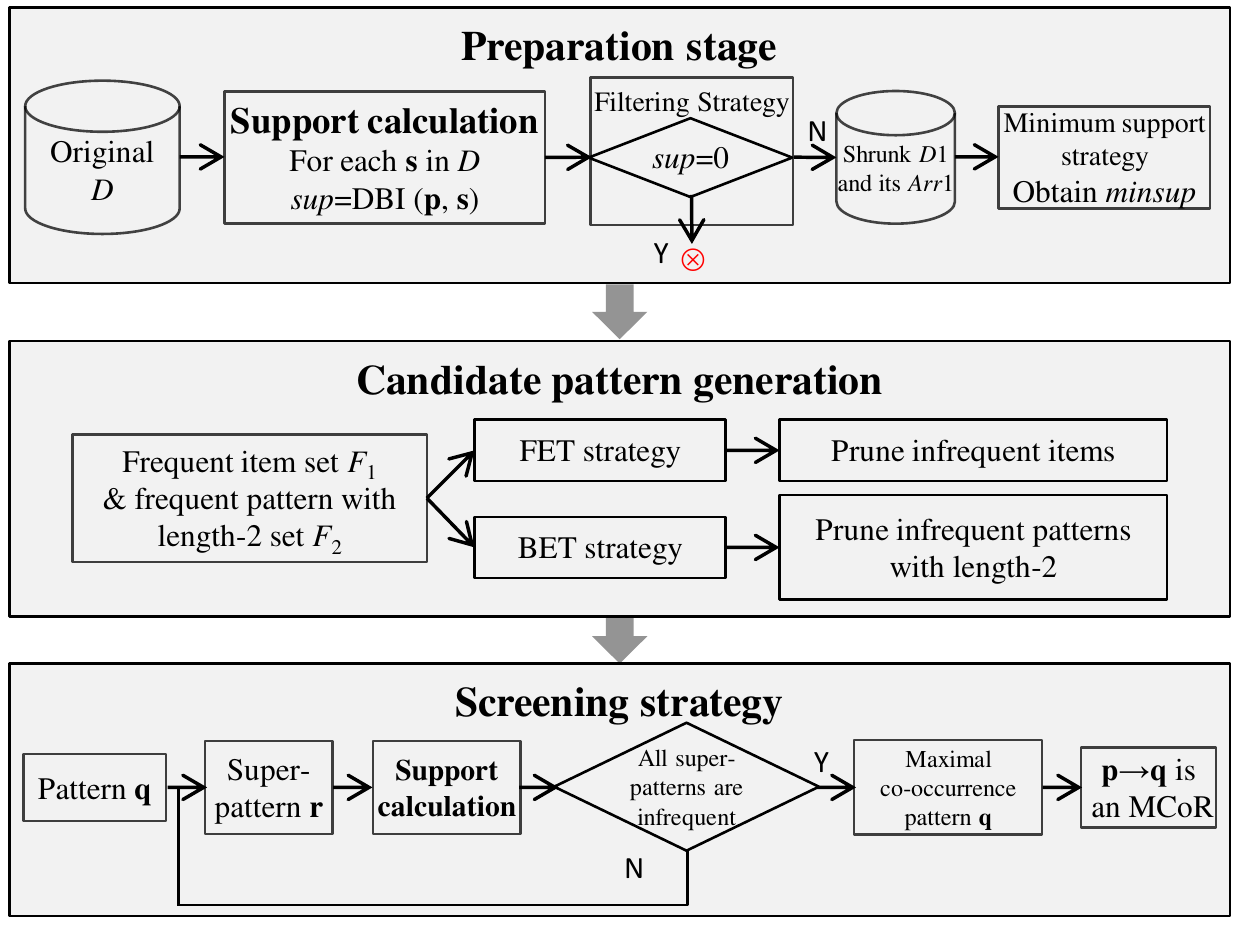"}
    \caption{Framework of MCoR-Miner. In the preparation stage, shown in  Section \ref{subsection: filtering stage}, a filtering strategy is proposed to prune the patterns whose support is zero to shrink the database, and a minimum support strategy is then explored to obtain $minsup$ based on the user-defined parameter $minconf$. In candidate pattern generation, shown in Section \ref{subsection: Candidate pattern generation}, a frequent item enumeration tree (FET) strategy and then a binomial enumeration tree (BET) strategy are developed to reduce the number of candidate patterns. Finally, a screening strategy shown in Section \ref{subsection: Screening strategy} is proposed to discover the maximal co-occurrence patterns to avoid a brute-force search. }
    \label{Framework}
\end{figure}

\subsection{Support calculation}
\label{subsection:Support calculation}

Given a sequence and a pattern with gap constraints, the calculation of its support is a pattern matching task \cite {netlap}. From Definition \ref{definition4}, we know that all nonoverlapping occurrences are a subset of all occurrences, and these occurrences can be expressed using a Nettree structure. Wu et al. \cite{netlap} first theoretically proved that calculating the nonoverlapping occurrences can be solved in polynomial-time. Some state-of-the-art algorithms, such as NETLAP-Best \cite {netlap} and NETGAP \cite {wu2018tcyb}, initially create a Nettree and then iteratively prune useless nodes to find all nonoverlapping occurrences on the Nettree. The time complexities of NETLAP-Best and NETGAP are both $O(m \times m \times n \times w)$, where $m$, $n$, and $w$ are the length of pattern and sequence, and $b-a+1$, respectively. Since pruning these useless nodes will consume a lot of time, Netback \cite {li2022apinm} was proposed to improve the efficiency, in which a Nettree is first created and a backtracking strategy is then employed to find all nonoverlapping occurrences on the Nettree. The time complexity of NetBack is reduced to $O(m \times n \times w)$. Although a Nettree can intuitively represent all occurrences, it contains a lot of useless information when representing all nonoverlapping occurrences, such as useless nodes and parent-child relationships. To further improve the efficiency, DFOM \cite {dfom2021kbs} was proposed; this algorithm does not need to create a whole Nettree, and employs depth-first search and backtracking strategies to find all nonoverlapping occurrences. One of the shortcomings of DFOM is that it employs a sequential searching strategy to find the feasible child nodes of each current node. The time complexity of DFOM is  $O(m \times n)$, since DFOM also employs the depth-first and backtracking strategies without creating a whole Nettree. To overcome this drawback, we propose a depth-first search and backtracking with indexes algorithm, called DBI, which employs an indexing mechanism to avoid sequential searching. Example \ref{example6} illustrates the principle of DFOM \cite {dfom2021kbs}.


\begin{example}\label{example6}
We use the same sequence \textbf{s} = $adbdadcdccabadcd$ as in Example \ref{example3} and a pattern \textbf{p} = $a[0,3]d[0,3]c$. Fig. \ref{DFOM} illustrates the occurrences searching process of DFOM. We know that $p_1$ = $a$. Thus, DFOM sequentially searches for $a$ in sequence \textbf{s}. Now, DFOM creates a root, labeled node 1, since $s_1$ = $p_1$ = $a$. According to the depth-first search strategy, DFOM sequentially searches for $d$ in sequence \textbf{s} after node 1 with gap constraints $[0,3]$, since $p_2$ = $d$. We know that $s_2$ = $d$. Thus, DFOM creates a child of node 1, labeled node 2. Since $p_3$ = $c$, DFOM sequentially searches for $c$ in sequence \textbf{s} after node 2 with gap constraints $[0,3]$. Unfortunately, there is no $c$ between $s_3$ and $s_6$. Hence, using a backtracking strategy, DFOM backtracks node 1 to find a new child. We know that $s_3$ is $b$, which is not equal to $p_2$ = $d$, and DFOM, therefore, continues to search. We know that $s_4$ is $d$ which is equal to $p_2$ = $d$. Thus, DFOM finds a new child of node 1, labeled node 4. Following a depth-first search strategy, DFOM searches for $c$ in sequence \textbf{s} after node 4 with gap constraints $[0,3]$, since $p_3$ = $c$. It is easy to see that $s_5$ and $s_6$ are not equal to $c$. Since $s_7$ = $c$ is equal to $p_3$ and the gap between positions 7 and 4 is 2, this satisfies the gap constraints $[0,3]$. Thus, DFOM finds a new child of node 4, labeled node 7. Since the length of \textbf{p} is three, DFOM finds an occurrence $<$1,4,7$>$. Now, DFOM continues to search for a new root. Since $s_2$, $s_3$, and $s_4$ are not equal to $a$, these characters are ignored. Since $s_5$ is $a$, DFOM finds a new root, labeled node 5. By iterating the above process, DFOM finds a new nonoverlapping occurrence $<$5,6,9$>$. Finally, occurrence $<$11,14,15$>$ is found.
\begin{figure}
    \centering    \includegraphics[width=0.7\linewidth]{"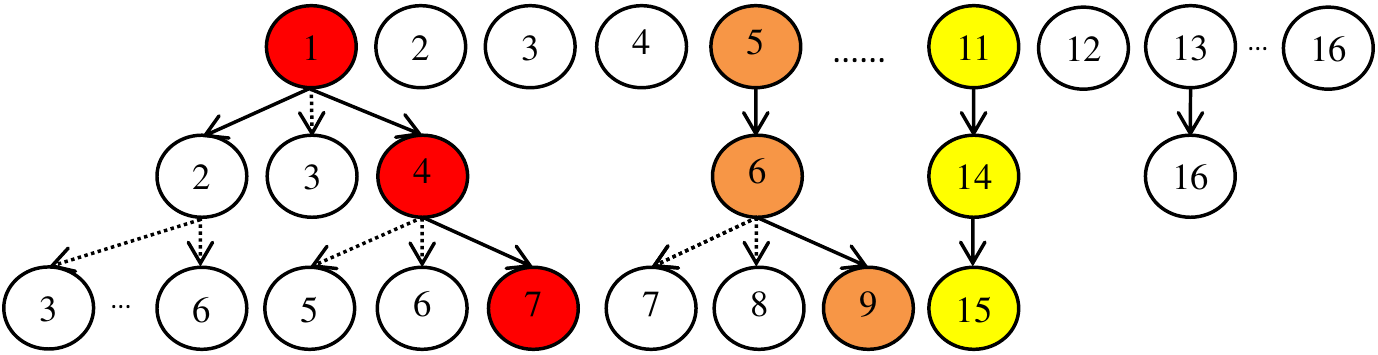"}
    \caption{Occurrences searching process of DFOM}
    \label{DFOM}
\end{figure}

\end{example}

From Example \ref{example6} and Fig. \ref{DFOM}, we know that the main drawback of DFOM is its sequential search. To overcome this shortcoming, DBI uses index arrays to store the positions of each character. Example \ref{example7}  illustrates the principle of DBI.

\begin{figure}
    \centering    \includegraphics[width=0.6\linewidth]{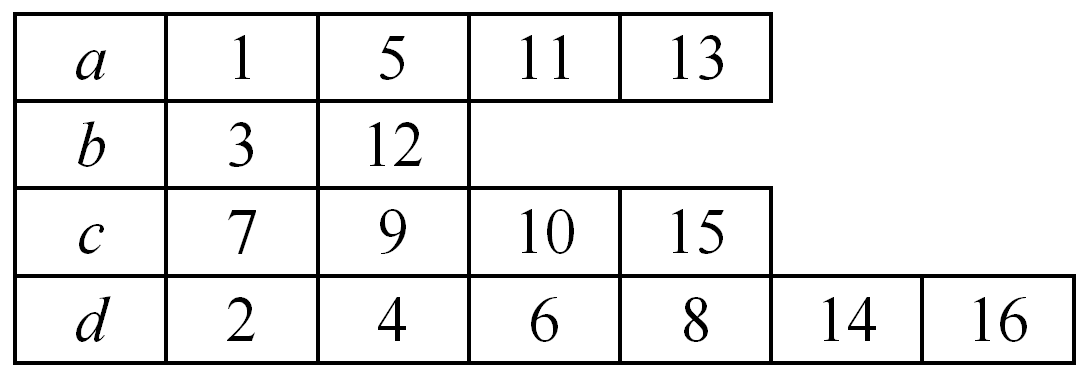}
    \caption{Index arrays for sequence \textbf{s}.}
    \label{The index arrays of sequence}
\end{figure}


\begin{example}\label{example7}
We use the same sequence  \textbf{s} = $adbdadcdccabadcd$ and pattern \textbf{p} = $a[0,3]d[0,3]c$ as in Example \ref{example6}.  DBI uses index arrays to store the positions of each character, as shown in Fig. \ref{The index arrays of sequence}. Fig. \ref{DBI} shows the occurrences searching process of DBI.  Since $p_1$ = $a$, DBI gets the first element in the array $a$, which is position 1. Thus, DBI creates a root, labeled node 1. Following a depth-first search strategy, since $p_2$ = $d$, DBI gets the first element in the array $d$, which is 2, and the gap between positions 2 and 1 is zero, which satisfies the gap constraints $[0,3]$. Thus, DBI creates a child of node 1, labeled node 2. According to a depth-first search strategy, since $p_3$ = $c$, DBI gets the first element in the array $c$, which is position 7. The gap between positions 7 and 2 is four, which does not satisfy the gap constraints $[0,3]$. According to the backtracking strategy, DBI backtracks node 1 to find a new child. We know that the second element of the array $d$ is position 4, which satisfies the gap constraints $[0,3]$. Thus, DBI creates a child of node 1, labeled node 4. Now, DBI selects the first element of the array $c$, which is position 7, and it satisfies the gap constraints $[0,3]$. Hence, DBI obtains the occurrence $<$1,4,7$>$. Then, DBI gets the second element of array $a$, which is position 5. By iterating the above process, DBI obtains the occurrence $<$5,6,9$>$. Finally, $<$11,14,15$>$ is found.
\begin{figure}
    \centering    \includegraphics[width=0.5\linewidth]{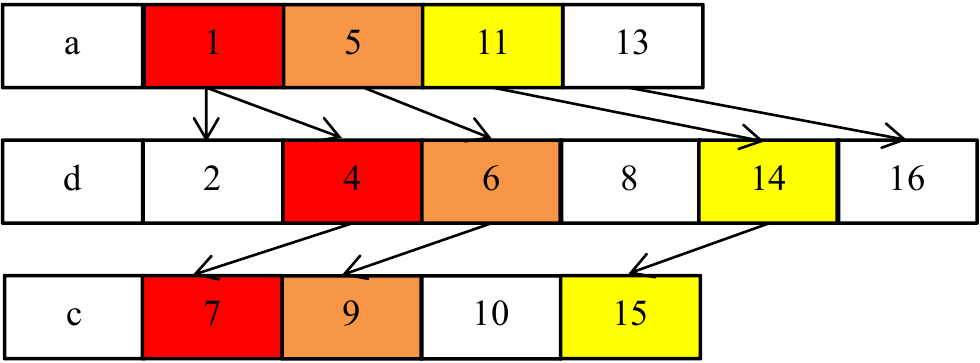}
    \caption{Occurrences searching process of DBI}
    \label{DBI}
\end{figure}
 
\end{example}

From Example \ref {example7}, we know that DBI first creates $m$ level nodes according to the index arrays of pattern \textbf{p}. Then, based on the index arrays, DBI adopts the depth-first search and backtracking strategies to iteratively find the minimal nonoverlapping occurrences. The pseudocode of DBI is shown in Algorithm 1. The main steps of DBI  are as follows.



\begin{enumerate}[Step 1:]

\item DBI selects the first element in the index array of $p_1$ as the current node, and set the current level as the first level (Lines 1 to 3).

\item DBI gets the first unused element in the next level as the child node of the current node, and the distance between the current node and the child node satisfies the gap constraints $gap$=[$a,b$] (Lines 5 to 10).  

\item If DBI successfully finds the child node, then the child node is selected as the current node and the next level is set as the current level. Otherwise, DBI backtracks to the parent node of the current node and searches for the next child node  (Lines 11 to 17).

\item Iterate Step 3 until DBI reaches the $m$-th level or the first level. If DBI reaches the $m$-th level, then DBI finds an occurrence, i.e., $sup(\textbf{p},\textbf{s})$++  (Lines 19 to 21).

\item DBI selects the next element in the index array of $p_1$ as the current node, and sets the current level as the first level.

\item Iterate Steps 2 to 5 until all elements in the index array of $p_1$ are checked.

\end{enumerate}

	
\begin{algorithm}[htb]	
\label{Algorithm 1}
\caption{DBI}
\begin{algorithmic}[1]
\REQUIRE  Index arrays \textbf{arr} for sequence \textbf{s}, pattern \textbf{p}, and $gap$=$[a,b]$
\ENSURE \textit{sup}(\textbf{p},\textbf{s})
\FOR {$root$ $\gets$ each element in \textbf{arr}[\textbf{p}[1]]}
    \STATE $CurNode$ $\gets$ \textbf{occ}[1] $\gets$ $root$;
    \STATE $CurLevel$ $\gets$ 1;
	\WHILE {$CurLevel>0$ $\&\&$ $CurLevel<len(\textbf{p})$}   

\STATE  $child$ $\gets$ next element in \textbf{arr}[\textbf{p}[$CurLevel+1]]$;  
\STATE  $gap$ $\gets$ $child-CurNode-1$;

\WHILE {$gap<a$}
       \STATE $child$ $\gets$  next element in \textbf{arr}[\textbf{p}[$CurLevel+1]]$; 
        \STATE  $gap$ $\gets$ $child-CurNode-1$;
   \ENDWHILE
		\IF {$a \leq gap$ $\&\&$ $gap \leq b$}
    
    \STATE $CurLevel$++;  

\STATE $CurNode$ $\gets$ \textbf{occ}$[CurLevel]$ $\gets$  $child$;

   \ELSE 
\STATE $CurLevel--$; 
\STATE $CurNode$ $\gets$ \textbf{occ}$[CurLevel]$;
	    \ENDIF
    \ENDWHILE
  	\IF{$CurLevel$ = len(\textbf{p})}
    \STATE $sup(\textbf{p},\textbf{s})++$;
    \ENDIF
  \ENDFOR
\RETURN $sup(\textbf{p},\textbf{s})$
\end{algorithmic}
\end{algorithm}


\begin{theorem}\label {theorem:timecomplexityDBI}
The time and space complexities of DBI are both $O(m \times n/h)$, where $m$, $n$, and $h$ are the lengths of pattern and sequence, and the size of  $\sum$, respectively.
\end{theorem}
\begin{proof}
It is easy to know that the size of the index array of each item is $O(n/h)$, since the length of the sequence is $n$, and the size of items is $h$. To calculate the support, there are $m$ index arrays. We know that each element in the index arrays can be used at most once, which means that DBI visits these elements one by one. Hence, the time and space complexities of DBI are both $O(m \times n/h)$. 
\end {proof}


\subsection{Preparation  stage}\label{subsection: filtering stage}

In this section, we propose two strategies: filtering strategy and minimum support strategy to avoid some redundant support calculations. We first propose the filtering strategy. 

\textbf {Filtering strategy}. If the support of pattern \textbf{p} in sequence \textbf{s} is zero, then sequence \textbf{s} can be pruned.

\begin{theorem} \label {theorem:filter}
Filtering strategy is correct and complete.
\end{theorem}
\begin{proof}
We know that the nonoverlapping SPM satisfies anti-monotonicity \cite {wu2018tcyb}, which means that the support of a superpattern is no greater than that of its sub-pattern. If $sup(\textbf{p},\textbf{s})$ = 0, then we can safely say that $sup(\textbf{q},\textbf{s}) $ = 0, where pattern \textbf{q} is a superpattern of pattern \textbf{p}, since the nonoverlapping SPM satisfies the anti-monotonicity condition and sequence \textbf{s} can therefore be pruned. Hence, the filtering strategy is correct and complete.
\end {proof}

Moreover, based on this anti-monotonicity, we further propose a minimum support strategy. 


\textbf{Minimum support strategy}. If the support of a pattern \textbf{q} in $D$ is less than $minsup$, then pattern \textbf{q} and its superpatterns can be pruned, where $minsup$ = $sup(\textbf{p}, D) $ $\times $ $mincf$. 

\begin{theorem} \label {theorem:minimum}
Minimum support strategy is correct and complete.
\end {theorem}
\begin{proof}
Suppose a rule \textbf{p} $\to$ \textbf{r} is a strong co-occurrence rule, i.e., $conf($\textbf{p} $\to$ \textbf{r}) = $sup(\textbf{p}$·$\textbf{r},D)$/$sup(\textbf{p}, D)$ $\geq$ $mincf$. Then, $sup(\textbf{p}$·$\textbf{r},D) \geq sup(\textbf{p}, D)$ $\times $ $mincf$, which means that the support of pattern $\textbf{p}$·$\textbf{r}$ is greater than or equal to $sup(\textbf{p}, D) \times mincf$. Thus, for MCoR mining, users do not need to set $minsup$, since it can be automatically calculated based on $mincf$. Hence, the minimum support strategy is correct and complete.
\end{proof}

In the preparation stage, we propose the SDB-Filt algorithm whose pseudocode is shown in Algorithm \ref{Algorithm 2}. The main steps of SDB-Filt are as follows.

\begin{enumerate}[Step 1:]

\item Select a sequence \textbf{s} in $D$.

\item SDB-Filt employs the DBI algorithm to calculate $sup(\textbf{p},\textbf{s})$.

\item If $sup(\textbf{p},\textbf{s})$ is zero, then sequence \textbf{s} can be shrunk according to filtering strategy. Otherwise, sequence \textbf{s} is stored in a shrunk database $D1$, and the support is updated, i.e., $sup(\textbf{p},D)$ gets $sup(\textbf{p},D)$ + $sup(\textbf{p},\textbf{s})$  (Lines 4 to 7).

\item Iterate Steps 1 to 3 until all sequences are calculated.

\item After getting $sup(\textbf{p},D)$, SDB-Filt calculates $minsup$ according to the minimum support strategy  (Line 9).
\end{enumerate}

\begin{algorithm}[htb]
\caption{SDB-Filt}\label{Algorithm 2}
\begin{algorithmic}[1]
 \REQUIRE  Sequence index arrays \textbf{Arr}, pattern \textbf{p}, $gap$=$[m,n]$, and $mincf$
	\ENSURE $minsup$ and shrunk index arrays $Arr1$
  
\STATE $sup$(\textbf{p},$D$) $\gets$ 0;
\FOR {each sequence \textbf{s} in $D$}
    \STATE $support$ $\gets$ DBI(\textbf{arr$_{\textbf{s}}$},\textbf{p}, $gap$);\\     
    \IF{$support \neq 0$} 
    \STATE  Store sequence index arrays \textbf{s} in $Arr1$;\\
   \STATE  $sup($\textbf{p},$D)$ $\gets$ $sup($\textbf{p},$D)$ $+support$;
    \ENDIF
\ENDFOR

  \STATE $minsup$ $\gets$ $sup(\textbf{p},D)$ $\times$  $mincf$;
  \RETURN $minsup$, $Arr1$
\end{algorithmic}
\end{algorithm}

\subsection{Candidate pattern generation}\label{subsection: Candidate pattern generation}

There are two strategies that are commonly used for generating candidate patterns: the enumeration tree strategy and the pattern join strategy. Many nonoverlapping SPM methods have adopted the pattern join strategy to generate candidate patterns, such as NOSEP \cite {wu2018tcyb} and NTP-Miner \cite {wu2021ntpm}, since it can effectively reduce the candidate patterns compared to the enumeration tree strategy. However, the pattern join strategy uses all frequent patterns with length $m$ to generate candidate patterns with length $m+1$, which means that we need to mine all of the frequent patterns. Obviously, a brute-force algorithm is used to mine all frequent patterns and then pattern \textbf{q} and its superpatterns are found among them. We therefore adopt the enumeration tree strategy. In the classical enumeration tree strategy, if pattern \textbf{q} is a frequent pattern, then all patterns $\textbf{q}$·$a$ are candidate patterns, where $a$ is any item in sequence \textbf{s}. 


To reduce the number of candidate patterns, we first propose the FET strategy and then present the BET strategy. 

\textbf{The FET strategy:} If pattern \textbf{q} is a frequent pattern and $y$ is any frequent item in sequence \textbf{s}, then all patterns $\textbf{q}$·$y$ are candidate patterns.

Example \ref{example8} demonstrates the advantage of the FET strategy.

\begin{example}\label{example8}
Suppose we have a sequence \textbf{s} = $s_1s_2s_3$$s_4s_5s_6$ $s_7s_8s_9s_{10}$ $s_{11}s_{12}$$s_{13}s_{14}$$s_{15}s_{16}$ = $adbdadcdc$$cabadcd$, a pattern \textbf{q} = $a[0,3]d[0,3]c$, and $minsup$ = 3. From sequence \textbf{s}, we know that all of the items are \{$a,b,c,d$\}. According to the classical enumeration tree strategy, there are four candidate patterns based on pattern \textbf{q}: $a[0,3]d[0,3]c[0,3]a$, $a[0,3]d[0,3]c[0,3]b$, $a[0,3]d[0,3]c[0,3]c$, and $a[0,3]d[0,3]c[0,3]d$. However, we know that the supports of patterns $a$, $b$, $c$, and $d$ are four, two, four, and six, respectively. Thus, the frequent items are \{$a,c,d$\}, since $minsup$ = 3. Hence, according to FET, there are only three candidate patterns based on pattern \textbf{q}: $a[0,3]d[0,3]c[0,3]a$, $a[0,3]d[0,3]c[0,3]c$, and $a[0,3]d[0,3]c[0,3]d$. This example shows that the FET strategy has better performance than the classical enumeration tree strategy. 
\end{example}

To further reduce the number of candidate patterns, inspired by the pattern join strategy, we propose the BET strategy as follows. 

\textbf{The BET strategy:} We discover all frequent patterns with length two and store them in set $F_2$. If $x$ is the last item of \textbf{q}, and pattern $xy$ is a frequent pattern with length two, i.e. $xy$ $\in$ $F_2$, then  pattern \textbf{r} = $\textbf{q}$·$y$ is a candidate pattern.

Example \ref{example9} illustrates the advantage of the BET strategy.

\begin{example}\label{example9}
  We use the same scenario as in Example \ref{example8}. We know that the frequent items are \{$a,c,d$\}. We can then generate nine candidate patterns with length two \{$aa$, $ac$, $ad$, $ca$, $cc$, $cd$, $da$, $dc$, $dd$\} using the pattern join strategy. It is then easy to see that the frequent candidate patterns with length two are $<$ad$>$, $<$ca$>$, $<$cd$>$, $<$dc$>$, and $<$dd$>$. From Example \ref{example8}, we see that there are only three candidate patterns based on pattern \textbf{q} using the FET strategy: $a[0,3]d[0,3]c[0,3]a$, $a[0,3]d[0,3]c[0,3]c$, and $a[0,3]d[0,3]c[0,3]d$. Pattern $a[0,3]d[0,3]c[0,3]a$ cannot be pruned by the BET strategy, since pattern $c[0,3]a$ is a frequent pattern with length two, while pattern $a[0,3]d[0,3]c[0,3]c$ can be pruned by the BET strategy, since pattern $c[0,3]c$ is not a frequent pattern with length two. Similarly, pattern $a[0,3]d[0,3]c[0,3]d$ cannot be pruned. Thus, only two candidate patterns are generated by the BET strategy. This example demonstrates that the BET strategy outperforms the FET strategy. A comparison of the candidate patterns is given in Fig. \ref{Comparison of candidate patterns generated by different strategies}.

\begin{figure*}[h]
		\centering
		\subfigure[Enumeration tree of all items]{\includegraphics[width=0.32\textwidth]{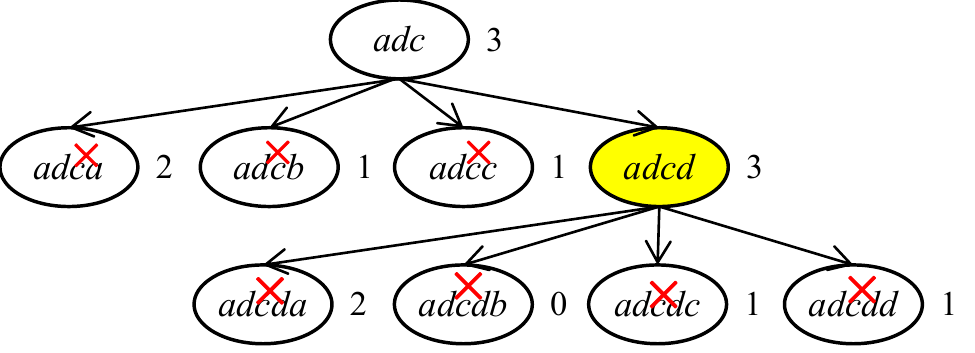}}
		\subfigure[Enumeration tree of FET]{\includegraphics[width=0.32\textwidth]{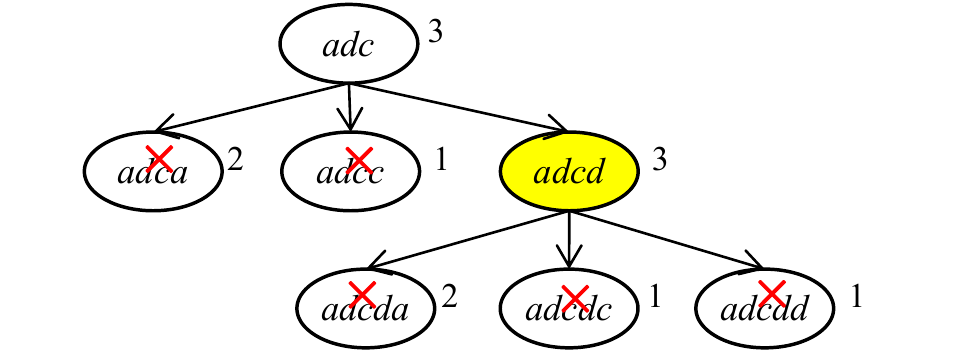}}
		\subfigure[Enumeration tree of BET]{\includegraphics[width=0.32\textwidth]{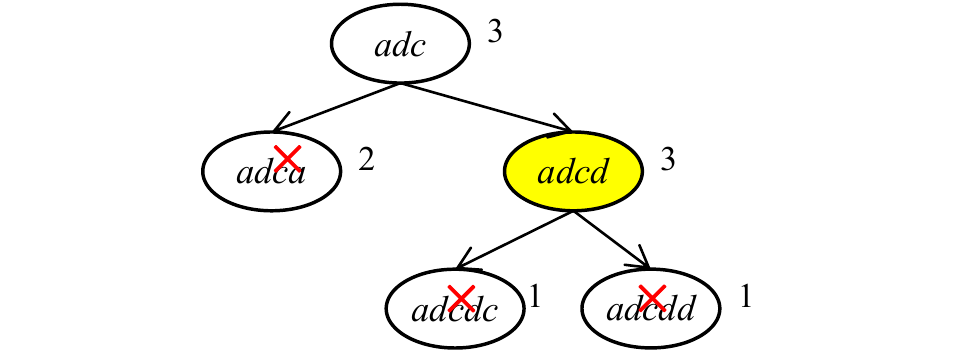}}
\caption{Comparison of candidate patterns generated by different strategies. Eight candidate patterns are generated by using all terms, while six and four candidate  patterns are generated by using the FET and BET strategies, respectively. }
		\label{Comparison of candidate patterns generated by different strategies}
	\end{figure*}
\end{example}

\begin{theorem} \label {theorem:fetbet}
Both FET and BET strategies are correct and complete.
\end {theorem}
\begin{proof}
As mentioned above, the nonoverlapping SPM satisfies anti-monotonicity \cite {wu2018tcyb}. Thus, the support of a superpattern is no greater than that of its sub-pattern. Therefore, if item $y$ is infrequent, then the pattern $\textbf{q}$·$y$ is infrequent. Hence, the FET strategy is correct and complete. Similarly, we know that the BET strategy is also correct and complete.
\end{proof}

\subsection{Screening strategy} \label{subsection: Screening strategy}

Obviously, it is easy for us to obtain strong co-occurrence sequential rules based on the frequent co-occurrence patterns. Moreover, we can discover the maximal strong co-occurrence rules based on the maximal co-occurrence patterns. The simplest method is to mine all frequent co-occurrence patterns, and then discover the maximal co-occurrence patterns among them. This method can be seen as a brute-force method. To improve the performance, we propose a screening strategy to discover the maximal co-occurrence patterns. The principle is as follows.

\textbf{Screening strategy:} Suppose we have a pattern \textbf{q}. If one of its superpatterns \textbf{r} is a frequent co-occurrence pattern, then pattern \textbf{q} is not a maximal co-occurrence pattern; otherwise, \textbf{q} is a maximal co-occurrence pattern. 

\begin{theorem}
Screening strategy is correct and complete.
\end{theorem}
\begin{proof}
According to the definition of the maximal co-occurrence pattern, for a pattern \textbf{q}, if one of its superpattern \textbf{r} is a frequent co-occurrence pattern, then pattern \textbf{q} is not a maximal co-occurrence pattern; otherwise, pattern \textbf{q} is a maximal co-occurrence pattern. The screening strategy discovers the maximal co-occurrence patterns based on the definition of the maximal co-occurrence pattern. Hence, the screening strategy is correct and complete.
\end{proof}

We store all frequent co-occurrence patterns in a stack. Suppose pattern \textbf{q} is the top element of the stack. We generate the candidate patterns based on pattern \textbf{q} using the BET strategy, and then we use the DBI algorithm to calculate the supports of the candidate patterns. If candidate pattern \textbf{r} is a frequent co-occurrence pattern, then we store pattern \textbf{r} in the stack; otherwise, pattern \textbf{q} is a maximal co-occurrence pattern and is stored in the set $M_p$. We check all patterns in the stack until it is empty. Example \ref{example10} illustrates the principle of the screening strategy. 

\begin{example}\label{example10}
    Suppose we have a sequence \textbf{s} = $s_1s_2s_3s_4$ $s_5s_6s_7s_8s_9s_{10}$$s_{11}s_{12}$$s_{13}s_{14}$$s_{15}s_{16}$ = $adbdadcdccabadcd$, \textbf{p} = $a[0,$ $3]d$, and $minsup$ = 3. Suppose the top element of the stack is pattern \textbf{q} = $a[0,3]d[0,3]c$. According to the BET strategy, there are two candidate patterns based on pattern \textbf{q}: $a[0,3]d[0,3]c[0,3]a$ and $a[0,3]d[0,3]c[0,3]d$. We know that the support of pattern $a[0,3]d[0,3]c[0,3]a$ is two. This means that $a[0,3]d[0,3]c[0,3]a$ cannot be stored in the stack, since it is not a frequent co-occurrence pattern. Moreover, we know that the support of pattern $a[0,3]d[0,3]c[0,3]d$ is three. The pattern $a[0,3]d[0,3]c[0,3]d$ is therefore stored in the stack, since it is a frequent co-occurrence pattern. Meanwhile, pattern \textbf{q} is not a maximal co-occurrence pattern, since it has a frequent co-occurrence superpattern $a[0,3]d[0,3]c[0,3]d$. Now, the top element of the stack is pattern $a[0,3]d[0,3]c[0,3]d$ which has two candidate patterns according to the BET strategy: $a[0,3]d[0,3]c[0,3]d[0,3]c$ and $a[0,3]d[0,3]c[0,3]d[0,3]d$. It is then easy to see that the supports of patterns $a[0,3]d[0,3]c[0,3]d[0,3]c$ and $a[0,3]d[0,3]c[0,3]d[0,3]d$ are both one. Hence, the pattern $a[0,3]d[0,3]c[0,3]d$ is a maximal co-occurrence pattern, since all its superpatterns are not frequent co-occurrence patterns. The advantage of the screening strategy is that the maximal co-occurrence patterns can be discovered without a brute-force search process.
\end{example}

\subsection{MCoR-Miner}

We introduce the MCoR-Miner algorithm whose pseudocode is shown in Algorithm 3.  MCoR-Miner has the following main steps.

\begin{enumerate}[Step 1:]
	\item Use SDB-Filt to calculate the support of pattern \textbf{p}, i.e., $sup(\textbf{p}, D)$; prune the sequence \textbf{s} whose support is zero; and calculate $minsup$ = $sup(\textbf{p}, D)$ $\times$ $mincf$  (Lines 1 to 2).

	\item  Traverse $D1$ to find all frequent items and store them in $F_1$ (Line 3).

	\item  Generate candidate pattern set $C_2$ using the pattern join strategy  (Line 4).

	\item  Calculate the support of each pattern in set $C_2$ and store the frequent patterns in $F_2$ (Line 5).

	\item  Store pattern \textbf{p} in stack $T$ (Line 6).

    \item Obtain the top element pattern \textbf{q} of stack $T$ (Line 8).

    \item Generate all candidate co-occurrence patterns of pattern \textbf{q} using frequent items $F_1$ and store them in set $C$ (Lines 10 to 13).

    \item Use BET to prune candidate patterns (Line 14).

    \item Use the DBI algorithm to calculate the support of each pattern \textbf{r} in $C$ (Line 15).

    \item If pattern \textbf{r} is a frequent co-occurrence pattern, then pattern \textbf{r} is stored in stack $T$ (Lines 16 to 19).

    \item If all patterns \textbf{r}  are not frequent co-occurrence patterns, then pattern \textbf{q} is a maximal co-occurrence pattern and is stored in set $M_p$ (Lines 22 to 24).

    \item Iterate Steps 6 to 11, until the stack is empty.

    \item Generate maximal strong rules based on $M_p$.
\end{enumerate} 	


\begin{algorithm}[htb]
	\label{Algorithm 3}
 \small
	\caption{ MCoR-Miner}
	\begin{algorithmic}[1]
\REQUIRE  Sequence database $D$, minimum confidence $mincf$, $gap$ = $[a,b]$, prefix pattern \textbf{p} 
		\ENSURE Maximal co-occurrence sequential rule set $M_r$\\
\STATE  Scan $D$ to obtain index arrays $Arr$; \\
//Use  the filtering strategy to shrink dataset, and the minimum support strategy to calculate $minsup$ 
\STATE $minsup$,$Arr1$  $\gets$ SDB-Filt($Arr$,\textbf{p},$gap$, $mincf$); \\
\STATE Scan $Arr1$ to find and store frequent items in $F_1$;
\STATE Generate candidate patterns with length two and store them in $C_2$; 
\STATE Discover frequent patterns with length two and store them in $F_2$;
        
\STATE Store pattern \textbf{p} in stack $T$;
          \WHILE{not empty $T$}
      \STATE { \textbf{q} $\gets$ $T$.pop(); }
     \STATE  {flag $\gets$ True;}\\
//Use the FET strategy to prune candidate patterns
  \FOR {each item $y$ in $F_1$} 
    \STATE $C$ $\gets$ $C$ $\cup$ \textbf{q}·$y$; 
    \STATE {$x$ $\gets$ the last item of $\textbf{q}$;} 
    \STATE { \textbf{r} $\gets$ \textbf{q}·$y$;}\\
  //Use the BET strategy to prune candidate patterns
\IF {pattern $x$·$y$ in $F_2$}
  \STATE $support$ $\gets$ DBI($Arr1$,\textbf{r},  $gap$); 
		 \IF {$support \geq minsup$}
        \STATE $T$.push (\textbf{r});
        \STATE flag $\gets$ False;
            \ENDIF
	    \ENDIF
       \ENDFOR \\
 //Use the screening strategy to mine the patterns
\IF{ flag=True}   
    \STATE $M_p$ $\gets$ $M_p$ $\cup$ \textbf{q};
    \ENDIF
  \ENDWHILE
   \STATE $M_r$ $\gets$ Generate maximal strong rules based on $M_p$;
  \RETURN $M_r$;
	\end{algorithmic}
\end{algorithm}

\begin{theorem}\label {theorem:timecomplexityDBI}
The space and time complexities of MCoR-Miner are $O(N)$ and $O(t \times M \times N/h)$, where $t$, $M$, $N$, and $h$ are the number of candidate patterns, the maximal length of candidate patterns, the total length of shrunk sequence database, and the size of  $\sum$, respectively.
\end{theorem}
\begin{proof}
The time and space complexities of scanning the sequence database $D$ to obtain its index arrays are $O(N)$. It is easy to know that the space complexity of SDB-Filt is $O(N)$. The time complexity of SDB-Filt is $O(m \times N/h)$, since the time complexity of DBI is $O(m \times n/h)$ and SDB-Filt checks all sequences, where $N=\sum _{i=1}^k n_i$. The space complexity of $F_1$ is $O(h)$, since the size of all frequent items is no greater than all items. Thus, the time complexity of Line 3 is $O(h)$. The time and space complexities of Line 4 are $O(h^2)$. The time complexity of Line 5 is $O(h^2 \times 2 \times N/h)$, since there are $O(h^2)$ candidate patterns, and the length of each candidate pattern is two. The space complexity of storing frequent patterns with length two is $O(h^2 \times 2)$. Suppose MCoR-Miner checks $t$ candidate patterns and the maximal length of candidate patterns is $M$. The space complexity of storing these candidate patterns is $O(t\times M)$, and the time complexity of calculating the supports of these candidate patterns is $O(t\times M\times N/h)$. Hence, the space complexity of MCoR-Miner is $O(N+N+h+h^2+k\times M)$ =$O(N)$, since generally $h$, $h^2$, $k$, $M$, and $k\times M$ are much smaller than $N$. The time complexity of MCoR-Miner is $O(N+ m \times N/h+h^2+ h^2\times 2\times N/h + t\times M\times N/h)$=$O(t\times M\times N/h)$, since $m$ and $2\times h^2$ are much smaller than $M$ and $t$, respectively.
\end {proof}



\section{Experimental results and analysis}
\label{section:Experimental results and analysis}

To validate the performance of MCoR-Miner, we consider the following eight research questions (RQs).

\begin{enumerate}[RQ1:]
\item Does MCoR-Miner yield better performance than other state-of-the-art algorithms in terms of mining co-occurrence patterns? 
\item Does DBI give better performance than other state-of-the-art algorithms in terms of support calculation? 
\item Can the filtering strategy reduce the number of support calculations in MCoR-Miner? 
\item Can the FET and BET strategies reduce the number of candidate patterns, and when can BET achieve better performance? 
\item What is the effect of the screening strategy? 
\item Does MCoR-Miner have better scalability? 
\item How do different parameters affect the performance of MCoR-Miner?
\item What is the performance of MCoR-Miner for recommendation tasks?
\end{enumerate}

To answer RQ1, we select NOSEP, DFOM-All, and DBI-All to verify the performance of MCoR-Miner on the problem of co-occurrence pattern mining. In response to RQ2, we propose NETGAP-MCoR, Netback-MCoR, and DFOM-MCoR to explore the effect of DBI, as described in Section \ref{subsection:Efficiency}. To address RQ3, we propose MCoR-NoFilt to investigate the effect of the filtering strategy, as presented in Section \ref{subsection:Efficiency}. To answer RQ4, we apply two algorithms, MCoR-AET and MCoR-FET, to verify the effects of the FET and BET strategies, as described in Section \ref{subsection:Efficiency}. To solve RQ5, we propose MCoR-NoScr to validate the effect of the screening strategy in Section \ref{subsection:Efficiency}. To answer RQ6, we explore the scalability of MCoR-Miner, as presented in Section \ref{subsection:scalability}. To address RQ7, we explore the effects of different gap constraints and minimum confidence thresholds on the running time of MCoR-Miner in Section \ref{subsection:Influence of parameters}. To answer RQ8, we report the performances of MCoR-Miner in terms of confidence, recall, and F1-score in section \ref{subsection:Comparison of confidence}. All experiments were conducted on a computer with an Intel(R) Core(TM)i7-10870H processor, 32 GB of memory, the Windows 10 operating system, and VS2022 as the experimental environment. All algorithms and datasets can be downloaded from https://github.com/wuc567/Pattern-Mining/tree/master/MCoR-Miner. 

\subsection{Benchmark datasets and baseline algorithms}

We used eight datasets shown in Table \ref{Description of datasets} as test sequences. 

\begin{table} [h]
\centering
    \scriptsize
    \caption{Description of datasets}
    \begin{tabular}{lccccc}
        \toprule
        \textbf{Dataset}  &   Type  & \textbf{$\sum$}    & \textbf{Number of sequence}          & \textbf{Length}   \\
        \midrule
        {GAMESALE\textsuperscript{1}}  & Commerce  &12  & 39  &16,446  \\
        {BABYSALE\textsuperscript{2}}  & Commerce  &6  & 	949  &29,971  \\
        {TRANSACTION\textsuperscript{3}}  &Commerce  &31  & 4,661  &32,851  \\		
        {MOVIE\textsuperscript{4}}  & Commerce  & 20  & 4,988  &741,132 \\		
        {MSNBC\textsuperscript{5}}  & Clickstream  &17  & 200,000  &948,674  \\		
        {SARS\textsuperscript{6}}  &  Bio-sequence  &4 & 426  &29,751  \\		
        {SARS-Cov-2\textsuperscript{7}}  & Bio-sequence  & 4  & 428  &29,903  \\		
        {HIV\textsuperscript{8}}  & Bio-sequence  & 20  & 1  &10,000  
        \\		   
        \bottomrule
    \end{tabular}
	\label{Description of datasets}
    
	\flushleft
 \scriptsize
	\begin{enumerate}[1:] 
		\item GAMESALE is a dataset of video game sales with ratings, which can be downloaded from http://www.kaggle.com/rush4ratio/video-game-sales-with-ratings.
		\item BABYSALE is a sales dataset for infant products, which can be downloaded from https://tianchi.aliyun.com/dataset/dataDetail?dataId=45.
		\item TRANSACTION is a dataset of sales behavior, which can be downloaded from https://www.kaggle.com/retailrocket/ecommerce-dataset.
		\item MOVIE contains users' movie ratings, which can be downloaded from http://grouplens.org/datasets/.
		\item MSNBC is a clickstream sequence database, which can be downloaded from http://www.philippe-fournier-viger.com/spmf/index.php?link=datasets.php.
		\item SARS contains the gene sequences of the virus that caused severe acute respiratory syndrome in 2003, and can be downloaded from https://www.ncbi.nlm.nih.gov/nuccore/30271926.
  	\item SARS-COV-2 contains the gene sequences of the virus causing COVID-19, and can be downloaded from https://www.ncbi.nlm.nih.gov/nuccore/MN908947.
   	\item HIV can be downloaded from http://archive.ics.uci.edu/ml/machine-learning-databases/00330/.
	\end{enumerate}
\end{table}

To report the performance of the MCoR-Miner algorithm, four experiments were designed, and nine competitive algorithms were selected. 

\begin{enumerate} [1.]

\item NOSEP \cite {wu2018tcyb}, DFOM-All  \cite {dfom2021kbs}, and DBI-All: To verify that it is necessary to use a special algorithm to discover frequent co-occurrence patterns, we selected two state-of-the-art algorithms: NOSEP \cite {wu2018tcyb} and DFOM \cite {dfom2021kbs}, where NOSEP can mine all frequent patterns. DFOM-All and DBI-All can also discover all frequent patterns, which employed the DFOM and DBI algorithms to calculate the support, respectively. We can then find the frequent co-occurrence patterns among all frequent patterns. Based on the above analysis, we know that the time complexities of NOSEP, DFOM-All, and DBI-All are $O(t_a \times M\times M\times N\times w)$, $O(t_a \times M\times N)$, and  $O(t_a \times M\times N/h)$, respectively, where  $t_a$ ($t_a \geq t$) and $w$ are the number of all frequent patterns and $b-a+1$.


\item  NETGAP-MCoR \cite {wu2018tcyb}, Netback-MCoR \cite{li2022apinm}, and DFOM-MCoR \cite {dfom2021kbs}:  To explore the running performance of DBI, three competitive algorithms were proposed: NETGAP-MCoR, Netback-MCoR, and DFOM-MCoR. The three competitive algorithms employed three state-of-the-art algorithms NETGAP \cite {wu2018tcyb}, Netback \cite{li2022apinm}, and DFOM  \cite {dfom2021kbs} to calculate the support, respectively. Other strategies, such as filtering, FET, BET, and screening, are the same as  MCoR-Miner. Based on the above analysis, we know that the time complexities of NETGAP-MCoR, Netback-MCoR, and DFOM-MCoR are $O(t \times M\times M\times N\times w)$, $O(t \times M\times N\times w)$, and $O(t \times M\times N)$, respectively.

\item  MCoR-NoFilt: To demonstrate the performance of the filtering strategy, MCoR-NoFilt was proposed, which did not employ the filtering strategy but is otherwise the same as MCoR-Miner.  The time complexity of MCoR-NoFilt is $O(t \times M\times N_o/h)$, while that of MCoR-Miner is $O(t \times M\times N/h)$, where $N_o$ and  $N$ ($N_o \geq N$) are the lengths of the original and shrunk sequence databases, respectively.
	
\item MCoR-AET and MCoR-FET: To investigate the performance of the BET strategy, we proposed MCoR-AET and MCoR-FET, which used all items and frequent items to generate candidate patterns, respectively.  The time complexities of MCoR-AET and MCoR-FET are $O(t_a \times M\times N/h)$ and $O(t_f \times M\times N/h)$, respectively, where $t_a$ ($t_a \geq t$) and $t_f$ are the numbers of candidate patterns. Note that $t_f$ can be smaller than $t$, since MCoR-Miner needs to calculate $O(h^2)$ candidate patterns with length two, and in some cases, it cannot prune $O(h^2)$ candidate patterns using the BET strategy.

\item MCoR-NoScr: To assess the performance of the screening strategy, we used MCoR-NoScr, which did not include the screening strategy. The time complexity of MCoR-NoScr is almost the same as that of MCoR-Miner.

\item CoP-Miner: To verify the recommendation performance,  CoP-Miner was proposed, which can mine all frequent patterns with the predefined prefix pattern, and employed the same strategies as MCoR-Miner. The time complexity of CoP-Miner is $O(t_c \times M\times N/h)$, where $t_c$ ($t_c \geq t$) is the number of all frequent patterns.
\end{enumerate} 

\subsection{Efficiency}
\label{subsection:Efficiency}

To determine the necessity of co-occurrence pattern mining, we selected NOSEP, DFOM-All, and DBI-All as the competitive algorithms. Moreover, to verify the performance of the DBI algorithm and the filtering, BET, and screening strategies, we employed NETGAP-MCoR, Netback-MCoR, DFOM-MCoR, MCoR-NoFilt, MCoR-AET, MCoR-FET, and MCoR-NoScr as the competitive algorithms. Eight datasets were selected. Since the characteristics of these datasets are significantly different, for example in terms of different number of characters and sequences, we set different parameters as shown in Table \ref{Parameters}. Comparisons of the running time and memory usage are shown in Tables \ref{Comparison of running time}  and \ref{Comparison of memory usage}, respectively. The main indicators of mining results, such as the number of CoRs and  MCoRs are shown in Table \ref{tab7:Main indicators of mining results}.


\begin{table}
\centering
     \scriptsize
    \caption{Parameters}
       \begin{tabular}{lccccc}
        \toprule
        {Dataset}     & {Prefix}    & {Gap constraint}          & {Minimum confidence}   \\
        \midrule
        {GAMESALE}  & $d$  & [0,8]  & 0.3  \\
        {BABYSALE}  & $a$  & 	[0,8]  &0.2  \\
        {TRANSACTION}   & $z$  & [0,3]  &0.3  \\		
        {MOVIE}   & $d$  & [0,3]  &0.5 \\		
        {MSNBC}    & $e$  & [0,8]  &0.3  \\		
        {SARS}    & $C$ & [0,3]  &0.6  \\		
        {SARS-Cov-2}  & $C$  & [0,3]  &0.6  \\		
        {HIV}  & $C$  & [0,3]  &0.3  \\		   
        \bottomrule
    \end{tabular}
	\label{Parameters}
\end{table}

\begin{table*}
	\centering
 \scriptsize
\caption{Comparison of running time (s)}
 	\label{Comparison of running time}
\begin{tabular}{ccccccccc}
		\toprule   
		&GAMESALE&BABYSALE&TRANSACTION&MOVIE&MSNBC&SARS&SARS-Cov-2& HIV\\
		\midrule 
		NOSEP &	47.95 & 152.55  & 10.20  & 77.66  & -
 &95.47  &133.12  &211.98 \\
		DFOM-All &	2.47 & 12.33  & 2.11  & 5.78  & - & 9.07 
& 10.00 & 84.28 \\
  DBI-All&	0.92 	&8.38 	&0.45 &4.13 &15.72 	&6.56 
	&8.28 	&1.48 \\
  NetGAP-MCoR&	16.17 	&113.00	&2.80&50.16 	&587.91
	&2.80 &2.75 &7.74\\
  Netback-MCoR&15.89&	112.17&	2.70 &	46.53&	582.67 &	2.58&	2.60 &	7.78 \\
  DFOM-MCoR&	1.88 &	11.34 &	0.70 &	5.25 &	97.38 &	3.60 &	3.13 &	5.11\\
  MCoR-NoFilt&	0.66 &	5.53 &	1.00 &	4.08 &	215.83 &	2.81 &	2.97 &	0.13 \\
  MCoR-AET&	0.70 &	7.14 &	3.88 &	16.22 &	25.92 &	1.56 &	1.42 &	\textbf{0.02} \\
  MCoR-FET&	0.59 &	5.94 &	0.47 &	5.25 &	24.17 &	\textbf{1.55} &	\textbf{1.41} &	\textbf{0.02} \\
 MCoR-NoScr&0.63 	&5.80 &	0.48 &	4.11 &	16.52 &	3.13 &	2.97 &	0.13 \\
  MCoR-Miner& \textbf{0.42} 	&\textbf{5.36} &	\textbf{0.38} &	\textbf{4.00} &	\textbf{16.08} &	2.34 &	2.66 &	0.11 \\
				\bottomrule
	\end{tabular}
\end{table*}

\begin{table*}
\centering
 \scriptsize

	\caption{Comparison of memory usage (Mb)}
 	\label{Comparison of memory usage}
\begin{tabular}{ccccccccc}
		\toprule   
		&GAMESALE&BABYSALE&TRANSACTION&MOVIE&MSNBC&SARS&SARS-Cov-2& HIV\\
		\midrule 
		NOSEP &	32.1  & 203.3   & 908.1  &1092.5   & 1183.2  &106.2   &106.5   &19.4  \\
		DFOM-All &	22.0  & 195.7   & 903.7   & 971.2  & 1188.0 & 95.1  & 95.4  & 20.2 \\
  DBI-All&	14.4  	&16.9  	&21.6  &38.5  &230.7 &17.6  &17.6  	&14.1 \\
  NetGAP-MCoR&	23.7  	&25.8 	&24.8 &163.1 	&404.9 
	&25.2&25.2&19.6\\
  Netback-MCoR&26.7&	18.8 &	35.7
 &	187.2&	263.9 &27.8&	27.4 &	20.3\\
  DFOM-MCoR&	15.9  &	18.7  &	22.8  &	40.0 
 &	226.0  &	15.8  &15.9 &	15.9 \\
  MCoR-NoFilt&	16.7  &	19.6 
&	25.0  &	41.1  &445.8  &	17.6  &	24.3 
 &	23.0  \\
  MCoR-AET&	16.0  &	18.8 
 &	24.5 &	43.9  &232.4  &	17.2 &	17.2  &16.0 
 \\
  MCoR-FET&	24.0  &26.4  &	30.5  &	47.8  &238.80 
 &	\textbf{14.9}  &	\textbf{14.8}  &	\textbf{13.4}  \\
 MCoR-NoScr&20.8 	&32.1 &30.3 &	47.4  &239.0 
 &	24.3  &16.9  &	15.8  \\
  MCoR-Miner&\textbf{13.6} &\textbf{16.5}  &\textbf{20.8} &	\textbf{38.0}  &\textbf{228.8}  &17.4 
 &	26.6 &	21.5 \\
				\bottomrule
	\end{tabular}
\end{table*}



\begin{table*}
	\centering
 \scriptsize
	\caption{Main indicators of mining results}
 
	\label{tab7:Main indicators of mining results}
	\begin{tabular}{ccccccccc}
		\toprule   
		&GAMESALE&BABYSALE&TRANSACTION&MOVIE&MSNBC&SARS&SARS-Cov-2& HIV\\
		\midrule 
		Number of CoRs &	36	& 98 & 4 & 4 & 12 &2 &2 &3\\
		Number of MCoRs &	20 & 47 & 2 & 1 & 1 & 2& 2& 3\\
  Number of filtered sequences&	642	&1,350	&104,787&	288	&49,163,184	&28	&28	&0\\
  Number of non-filtered sequences&	11,877	&425,700	&21,060	&79,520	&1,236,816	&11,900	&11,956	&472\\
  Number of all items&	12&	6&	30&	20&	17&	4&	4&	20\\
  Number of frequent items&	10&	5&	4&	3&	15&	4&	4&	20\\
  Number of BET items&	27&	15&	4&	2&	1&	13&	13&	286\\
  Number of BET items under non-filtered&	27&	15&	6&	2&	24&	13&	13&	286\\
  Number of candidate patterns
 pruned by BET&	149&	70&	9&	8&	168&	0&	0&	8\\
  Number of candidate patterns 
nonpruned by BET&	221	&425&	11&	7&	27&	12&	12&	72\\
				\bottomrule
	\end{tabular}
\end{table*}

These results give rise to the following observations.
\begin{enumerate} [1.]

\item Maximal co-occurrence rule mining can effectively reduce the number of rules. From Table \ref{tab7:Main indicators of mining results},  on BABYSALE, there are 98 co-occurrence rules, while there are 47 maximal co-occurrence rules. The same phenomenon can be found on all the other datasets. This is because as shown in Example \ref{example5}, both $ad$ $\to$ $c$ and $ad$ $\to$ $cd$ are rules, and $ad$ $\to$ $cd$ is a maximal rule, while rule $ad$ $\to$ $c$ is not a maximal rule. This means that MCoR mining can effectively reduce the number of rules.

\item It is necessary to explore the co-occurrence pattern mining, since MCoR-Miner is significantly faster than NOSEP \cite {wu2018tcyb} and DFOM-All  \cite {dfom2021kbs}, and in particular is faster than DBI-All. For example, on GAMESALE, NOSEP and DFOM-All take 47.95s and 2.47s, respectively, while MCoR-Miner takes 0.42s, meaning that MCoR-Miner is about 100 times faster than NOSEP, and six times faster than DFOM-All. Moreover, DBI-all takes 0.92s, which is slower than MCoR-Miner. These phenomena can be found on the other datasets. The reasons for this are twofold. Firstly, MCoR-Miner employs a more efficient algorithm, called DBI, to calculate the support compared to NOSEP and DFOM-All. As mentioned in Section \ref{subsection:Support calculation}, we know that NOSEP has to create a whole Nettree to calculate the support. However, a Nettree contains many redundant nodes and parent-child relationships. Although DFOM can overcome this disadvantage, it uses a sequential searching strategy to find feasible child nodes of the current node, which is an inefficient method. In contrast, DBI is equipped with an indexing mechanism to avoid sequential searching, meaning that MCoR-Miner is significantly faster than NOSEP and DFOM-All. Secondly, although both DBI-All and MCoR-Miner employ DBI to calculate the support, DBI-All runs slower than MCoR-Miner. The reason for this is that the three algorithms have to discover all of the patterns, and since many of these are not targeted patterns, all three algorithms have to filter out these useless patterns. The experimental results therefore indicate that it is necessary to explore the co-occurrence pattern mining. 

\item  MCoR-Miner runs faster than NETGAP-MCoR, Netback-MCoR, and DFOM-MCoR on all datasets, which indicates the efficiency of DBI. For example, from Table \ref{Comparison of running time}, we can see that on BABYSALE, MCoR-Miner takes 5.36s, while NETGAP-MCoR, Netback-MCoR, and  DFOM-MCoR take 113.00s, 112.17s, and 11.34s, respectively. Thus, MCoR-Miner is obviously faster than NETGAP-MCoR and Netback-MCoR, and is about twice as fast as DFOM-MCoR. The reasons are as follows. The four algorithms adopt different methods to calculate the support. NETGAP-MCoR and Netback-MCoR have to create the whole Nettree which contain many useless nodes and parent-child relationships. Therefore, their time complexities are higher than those of DFOM-MCoR and MCoR-Miner. Although DFOM-MCoR does not create the whole Nettree, DFOM adopts a sequential searching strategy to find feasible child nodes of the current node, which is inefficient. Hence, MCoR-Miner outperforms NETGAP-MCoR, Netback-MCoR, and  DFOM-MCoR.

\item  MCoR-Miner runs faster than MCoR-NoFilt on all datasets, which indicates the effectiveness of the filtering strategy. According to Table \ref{Comparison of running time}, on BABYSALE, MCoR-Miner takes 5.36s, while MCoR-NoFilt takes 5.53s. This phenomenon can also be found on all the other datasets. The reason for this is that some useless sequences are filtered out. For example, from Table \ref{tab7:Main indicators of mining results}, we see that 1350 sequences are filtered out when minimg all rules, which means that it is not necessary to calculate the support for these 1350 sequences, and hence the running performance is improved. The filtering strategy can reduce not only  the number of sequences, but also the number of frequent patterns with length two, which means that it can reduce the number of candidate patterns with length more than two. For example, on the TRANSACTION dataset, there are only four patterns under the filtering strategy, but six under the non-filtering strategy, meaning that more time is needed to discover all the rules. All in all, the filtering strategy is effective, and MCoR-Miner outperforms MCoR-NoFilt.

\item MCoR-Miner runs faster than MCoR-AET and MCoR-FET on GAMESALE, BABYSALE, TRANSACTION, MOVIE, and MSNBC, but slower than MCoR-AET and MCoR-FET on SARS, SARS-Cov-2, and HIV, which indicates that the BET strategy has certain advantages and limitations. We choose two typical datasets to illustrate the reason for this. 1) Taking GAMESALE as an example, we see that there are 12 items and 10 frequent items. Thus, there are 10$\times$ 10=100 candidate patterns with length two. After checking these 100 candidate patterns, we have 27 frequent patterns with length two. Using these 27 patterns, 149 candidate patterns are pruned by the BET strategy, and 221 candidate patterns are not according to Table \ref{tab7:Main indicators of mining results}. Thus, by adding 100 times support calculations for candidate patterns with length two, MCoR-Miner can avoid 149 times support calculations for candidate patterns. Hence, MCoR-Miner can improve the running performance on GAMESALE. 2) Taking HIV as another example, we see that there are 20 frequent items. Thus, there are 20$\times$20=400 candidate patterns with length two. After checking these 400 candidate patterns, we have 286 frequent patterns with length two. Of these 286 patterns, only eight candidate patterns are pruned by the BET strategy, and 72 are not   according to Table \ref{tab7:Main indicators of mining results}. Thus, by adding 400 times support calculations for candidate patterns with length two, MCoR-Miner can avoid only eight times support calculations for candidate patterns. The running time is therefore increased. Thus, the experimental results show that the BET strategy has some limitations on the Bio-sequence datasets. Hence, our methods, MCoR-Miner or MCoR-FET, can obtain better performance in all cases.

\item  MCoR-Miner runs faster than MCoR-NoScr on all datasets, which indicates the effectiveness of the screening strategy. For example, on TRANSACTION, MCoR-Miner requires 0.38s, while MCoR-NoScr requires 0.48s. This phenomenon can be found on all the other datasets. The reason for this is as follows. The difference between the two algorithms is that MCoR-Miner is equipped with the screening strategy, while MCoR-NoScr is not. Thus, the results indicate that the screening strategy is an effective way to reduce the running time, and hence MCoR-Miner outperforms MCoR-NoScr.

\item  From Table \ref{Comparison of memory usage}, we notice that MCoR-Miner consumes less memory when mining patterns. For example, on GAMESALE, MCoR-Miner consumes 13.6Mb memory, which is less than the other competitive algorithms. This result can be found on most of the datasets. In particular, MCoR-Miner consumes less memory than NOSEP and DFOM-All algorithms. This is because in order to calculate the support, NOSEP needs to create a whole Nettree, which contains a lot of useless information. Although DFOM overcomes this drawback, it uses the sequential searching strategy to find feasible child nodes of the current node, which is not efficient. However, MCoR-Miner uses the DBI algorithm, which is equipped with the indexing mechanism to avoid sequential searching. Hence, MCoR-Miner consumes less memory.
\end{enumerate} 

In summary, MCoR-Miner achieves better running performance and memory consumption performance than other competitive algorithms.

\subsection{Scalability}
\label{subsection:scalability}

To validate the scalability of MCoR-Miner, we selected MSNBC as the experimental dataset, since it was the largest dataset in Table \ref{Description of datasets}. Moreover, we  created MSNBC1, MSNBC2, MSNBC3, MSNBC4, MSNBC5, and MSNBC6, which were one, two, three, four, five, and six times the size of MSNBC, respectively. We set $gap$ = $[0,8]$, $mincf$ = 0.3, and the prefix pattern was $e$. Comparisons of running time and memory usage are shown in Figs. \ref{Comparison of running time with different dataset size} and \ref{Comparison of memory usage with different dataset size}, respectively.

\begin{figure}
    \centering
    \includegraphics[width=0.95\linewidth]{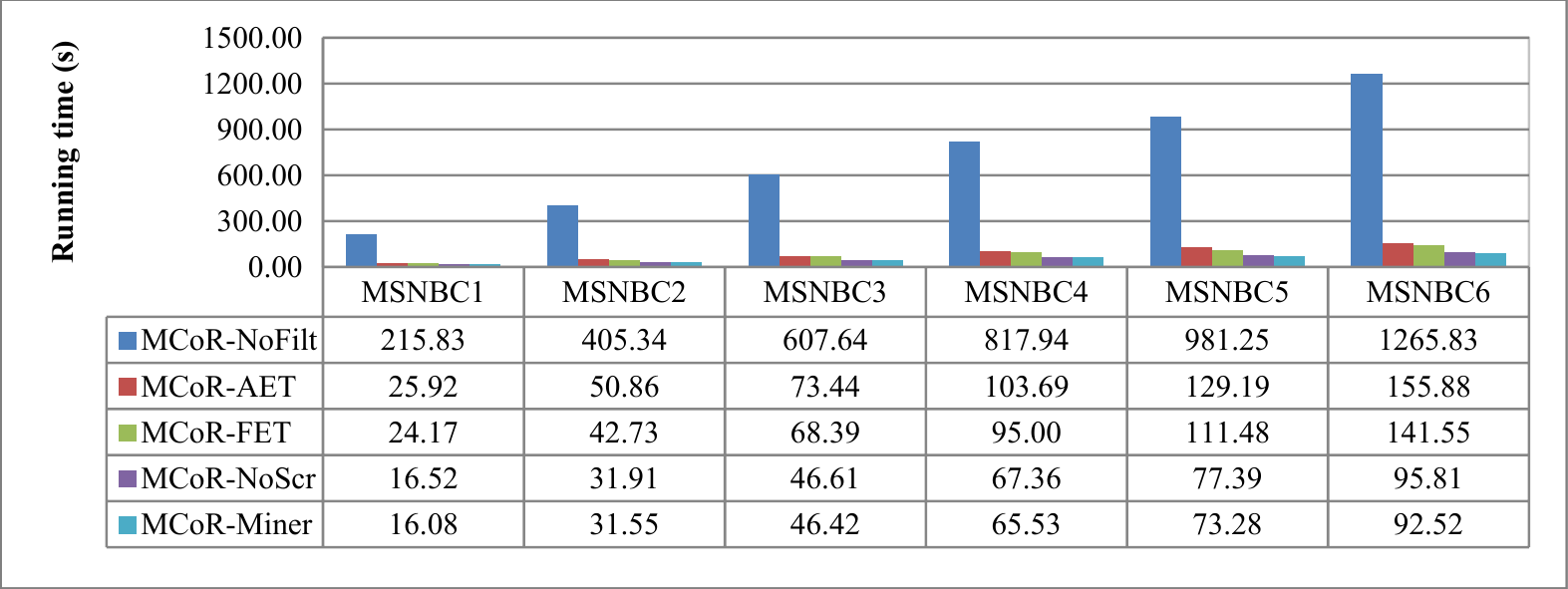}
    \caption{Comparison of running time with different dataset sizes (s) }
    \label{Comparison of running time with different dataset size}
\end{figure}

\begin{figure}
    \centering
    \includegraphics[width=0.95\linewidth]{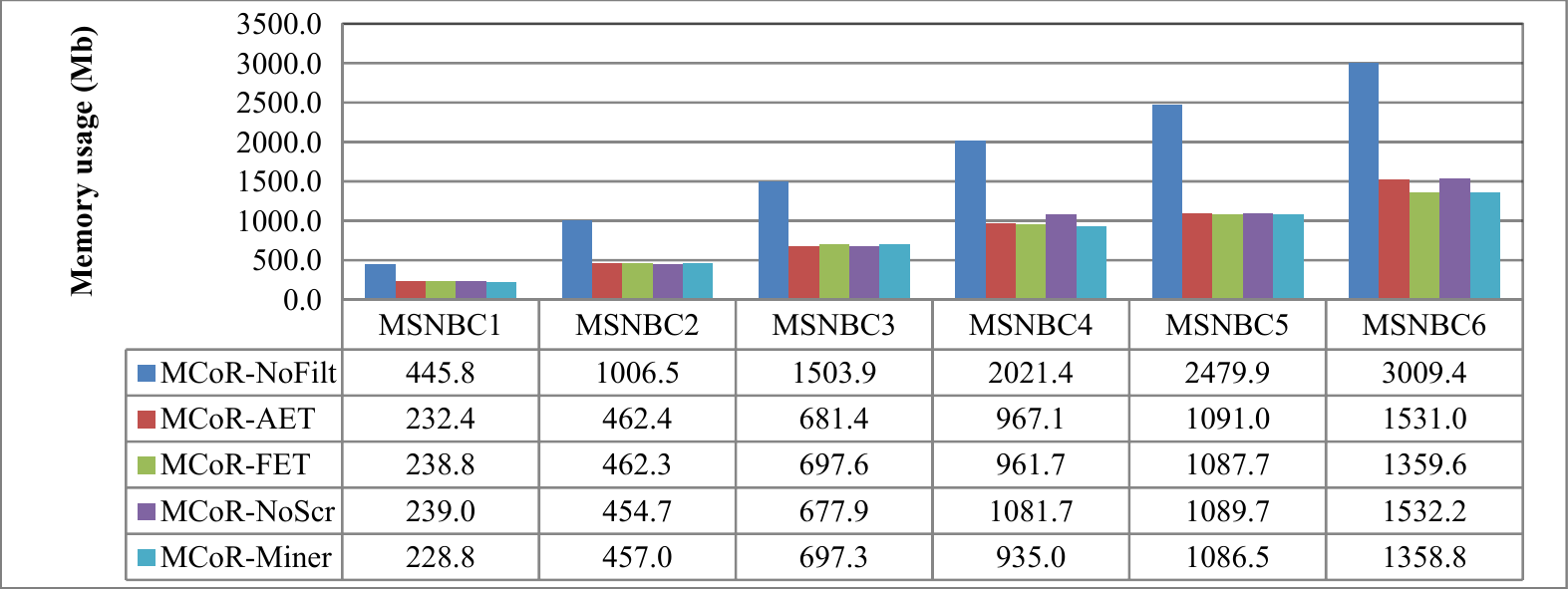}
    \caption{Comparison of memory usage with different dataset sizes (Mb)}
    \label{Comparison of memory usage with different dataset size}
\end{figure}

The results give rise to  the following observations. 

From Figs. \ref{Comparison of running time with different dataset size} and \ref{Comparison of memory usage with different dataset size}, we know that the running time and the memory usage of MCoR-Miner both show linear growth with the increase of the size of the dataset. For example, the size of MSNBC6 is six times that of MSNBC1. MCoR-Miner takes 16.08s on MSNBC1, and 92.52s on MSNBC6, giving 92.52/16.08=5.75. The memory usage of MCoR-Miner is 1358.8Mb on MSNBC6, and 228.8Mb on MSNBC1, giving 1358.8/228.8 = 5.94. Thus, the growth rates of the running time and memory usage are slightly lower than the increase of the dataset size. This phenomenon can be found on all the other datasets. These results indicate that the time and space complexities of MCoR-Miner are positively correlated with the dataset size. More importantly, we can see that MCoR-Miner is more than 10 times faster than MCoR-NoFilt, and that the memory usage of MCoR-NoFilt is more than twice that of MCoR-Miner. We therefore draw the conclusion that MCoR-Miner has strong scalability, since the mining performance does not degrade with the increase of the dataset size.

\subsection{Influence of parameters}
\label{subsection:Influence of parameters}

We assessed the effects of different gap constraints and minimum confidence on the number of rules and running time.

\subsubsection{Influence of different gap constraints} 

To analyze the influence of different gap constraints on the number of rules and running time of MCoR-Miner, we selected GAMESALE as the experimental dataset, and MCoR-NoFilt, MCoR-AET, MCoR-FET, and MCoR-NoScr as the competitive algorithms. The prefix pattern was $d$ and the minimum confidence was 0.3. The gap constraints were [0,5], [0,6], [0,7], [0,8], [0,9], and [0,10]. A comparison of the running time is shown in Fig. \ref{Comparison of running time with different gaps} and the numbers of mined CoRs and MCoRs mined are shown in Table \ref{Comparison of number of CoRs and MCoRs with different gaps}. 

\begin{figure}
    \centering
    \includegraphics[width=0.95\linewidth]{ 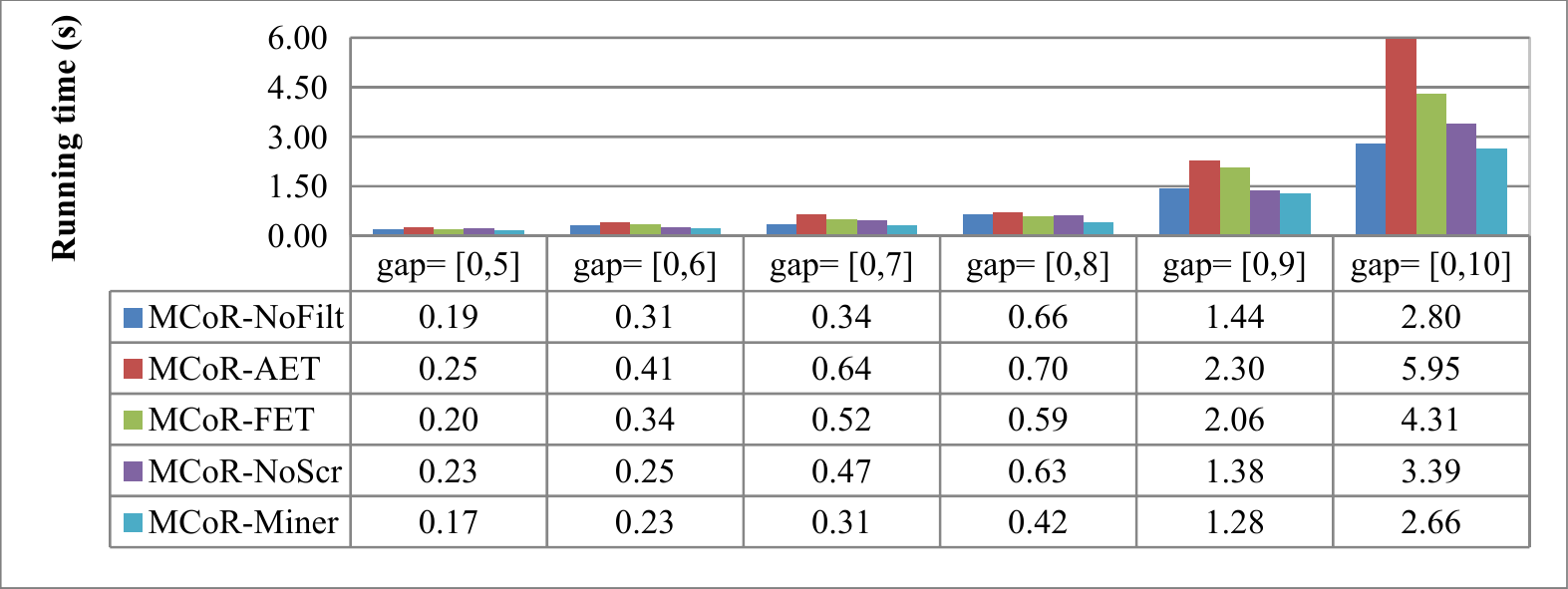}
    \caption{ Comparison of running time with different $gap$ (s)}
    \label{Comparison of running time with different gaps}
\end{figure}

\begin{table}
	\centering
	\caption{ Comparison of numbers of CoRs and MCoRs with different gaps}
	\begin{tabular}{ccccccc}
		\toprule   
		&gap=&gap=&gap=&gap=&gap=&gap=\\
		&[0,5]&[0,6]&[0,7]&[0,8]&[0,9]&[0,10]\\
		\midrule 
		Number of CoRs &	8	& 14 & 20 & 36 & 63 &134\\
		Number of MCoRs &	6 & 10 & 12 & 20 & 34 & 63\\
				\bottomrule
     \label{Comparison of number of CoRs and MCoRs with different gaps}
	\end{tabular}
\end{table}

These results give rise to  the following observations.

With the increase of the gap, the running time and number of CoRs and MCoRs are increased. For example, when gap=[0,8], MCoR-Miner takes 0.17s, and mines eight CoRs and six MCoRs, whereas when gap=[0,8], MCoR-Miner takes 0.42s, and mines 36 CoRs and 20 MCoRs. This phenomenon can be found on all the other algorithms. The reason is as follows. We know that with the increase of the gap, the support of a pattern increases. Thus, more patterns can be found, and each co-occurrence pattern corresponds to a CoR. Hence, more CoRs and MCoRs can be found and the running time also increases. MCoR-Miner outperforms the other competitive algorithms for any gap constraints.

\subsubsection{Influence of different minimum confidence}

To report the influence of different minimum confidences on the number of rules and running time of MCoR-Miner, we selected GAMESALE as the experimental dataset, and MCoR-NoFilt, MCoR-AET, MCoR-FET, and MCoR-NoScr as the competitive algorithms. The prefix pattern was $d$ and the gap constraints was [0,8]. We set the minimum confidence to 0.16, 0.20, 0.24, 0.28, 0.32, and 0.36. A comparison of the running time is shown in Fig. \ref{Comparison of running time with different mincf} and the numbers of CoRs and MCoRs mined are shown in Table \ref{tab:Comparison of number of CoRs and MCoRs with different mincf}.

\begin{figure}
    \centering
    \includegraphics[width=0.95\linewidth]{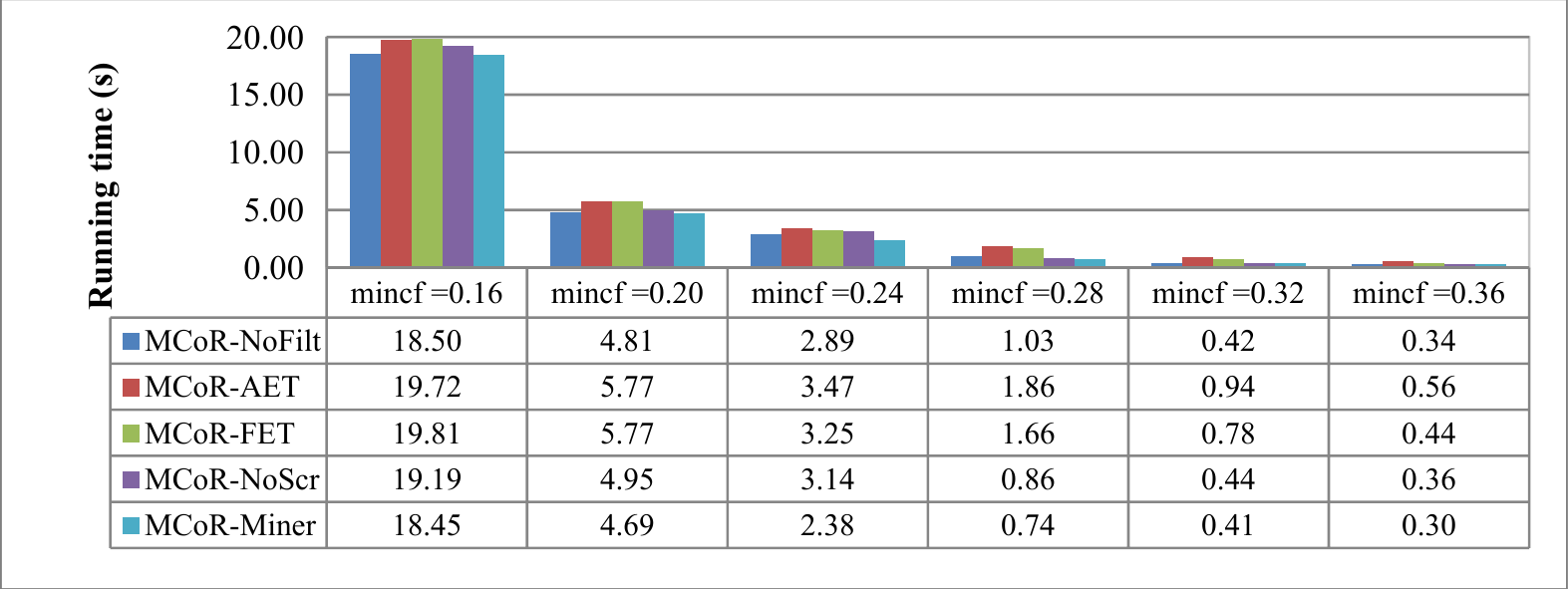}
    \caption{Comparison of running time with different $mincf$ (s)}
    \label{Comparison of running time with different mincf}
\end{figure}
\begin{table}
	\centering
	\caption{Comparison of numbers of CoRs and MCoRs with different $mincf$}
	\label{tab:Comparison of number of CoRs and MCoRs with different mincf}
	\begin{tabular}{ccccccc}
		\toprule   
&\textit {mincf}& \textit {mincf}& \textit {mincf}& \textit {mincf}& \textit {mincf}& \textit {mincf}\\
		&=0.16&=0.20&=0.24&=0.28&=0.32&=0.36\\
		\midrule 
		Number of CoRs &796	& 246& 100 & 49	& 27&17 \\
		Number of MCoRs &431 & 134 & 53 & 25 & 15 & 9\\
		\bottomrule
	\end{tabular}
\end{table}

The results give rise to the following observations.  

With the increase of minimum confidence, the running time and numbers of CoRs and MCoRs are decreased. For example, when $mincf$ = 0.16, MCoR-Miner takes 18.45s, and mines 796 CoRs and 431 MCoRs, whereas when $mincf$ = 0.36, MCoR-Miner takes 0.30s, and mines 17 CoRs and nine MCoRs. This phenomenon can be found on the other algorithms. The reason is as follows. For a certain prefix pattern, the minimum support threshold increases with the increase of $mincf$. Thus, few patterns can be frequent patterns, and each co-occurrence pattern corresponds to a CoR. Hence, fewer CoRs and MCoRs can be found and the running time also decreases. MCoR-Miner outperforms the other competitive algorithms for any minimum confidence.
	
\subsection{MCoR-Miner for recommendation}
\label{subsection:Comparison of confidence}

The recommendation task is to provide the recommendation schemes based on the discovered potential rules with the same given antecedent \textbf{p}. To investigate  the recommendation performance of MCoR-Miner,  we compare the recommendation performance of co-occurrence pattern mining and rule mining. Therefore, we employed CoP-Miner as the competitive algorithm and selected the commerce datasets: GAMESALE, BABYSALE, TRANSACTION, and MOVIE. The first 80\% of the sequences were used as the training set, and the remaining 20\% as the testing set. The parameters used for training are shown in Table \ref{Parameter setting}.

\begin{table}
    \centering
    \caption{Parameter settings}
    \begin{tabular}{ccccc}
        \toprule

\multirow{2}*
{Dataset}  & \multirow{2}*{Prefix } &  {Gap } & $mincf$ for  & $minsup$ for   \\
 &    & {constraint}    & MCoR-Miner          & CoP-Miner   \\
        \midrule
        {GAMESALE} &$d$& [0,3]	& 0.39	  & 400 \\
        {BABYSALE}  & $fc$  & 	[0,3]  &0.40 &300  \\
        {TRANSACTION}  & $xz$  & [0,3]  &0.40 &100  \\		
        {MOVIE}  & $dee$  & [0,3]  &0.40 &150000 \\		
        	   
        \bottomrule
    \end{tabular}
	\label{Parameter setting}
\end{table}

To evaluate the performance of recommendation based on the co-occurrence rules, in addition to confidence, there are three commonly used criteria: precision $Pr$ = $TP/(TP+FP)$, recall $Re$ = $TP/(TP+FN)$, and F1-score $F1$ = $2\times Pr\times Re/(Pr+Re)$, where $TP$ is the number of correct recommendations, $FP$ is the number of incorrect recommendations, and $FN$ is the number of relevant items that are not included in the recommendation list. Taking GAMESALE as an example, we see that there are two maximal co-occurrence rules in the training set for MCoR-Miner: \{$d$ $\to$ $d$, $d$ $\to$ $e$\}, which means that if a user purchases item $d$, then MCoR-Miner will recommend items $d$ and $e$ to the user. In the testing set, all of the patterns of length two  with the prefix pattern $d$ are $da$, $db$, $dd$, $de$, $dj$, and $dk$, with supports 110, 108, 139, 281, 126, and 83, respectively. Therefore, in this example, $TP$ is 420, since the sum of the supports for patterns $dd$ and $de$ is 139 + 281 = 420. FP is zero, since MCoR-Miner recommends items $d$ and $e$, and in the test set, the next-item actually can be $a$, $b$, $d$, $e$, $j$, and $k$, which means that items $d$ and $e$ are both in these items, i.e., MCoR-Miner has no error recommendation. $FN$ is 427, since the sum of the supports for patterns $da$, $db$, $dj$, and $dk$ is 110+ 108 + 126+ 83 = 427.  Thus, the precision, recall, and F1-score of MCoR-Miner on the testing set are 420/(420+0) = 1, 420/(420+427) = 0.4959, and 2$\times$1$\times$0.4959/(1+0.4959) = 0.6630, respectively. Comparisons of the confidence, recall, and F1-score results are shown in Figs. \ref{Comparison of confidence}, \ref{Comparison of recall}, and \ref{Comparison of F1-score}, respectively.

\begin{figure}
    \centering
    \includegraphics[width=0.95\linewidth]{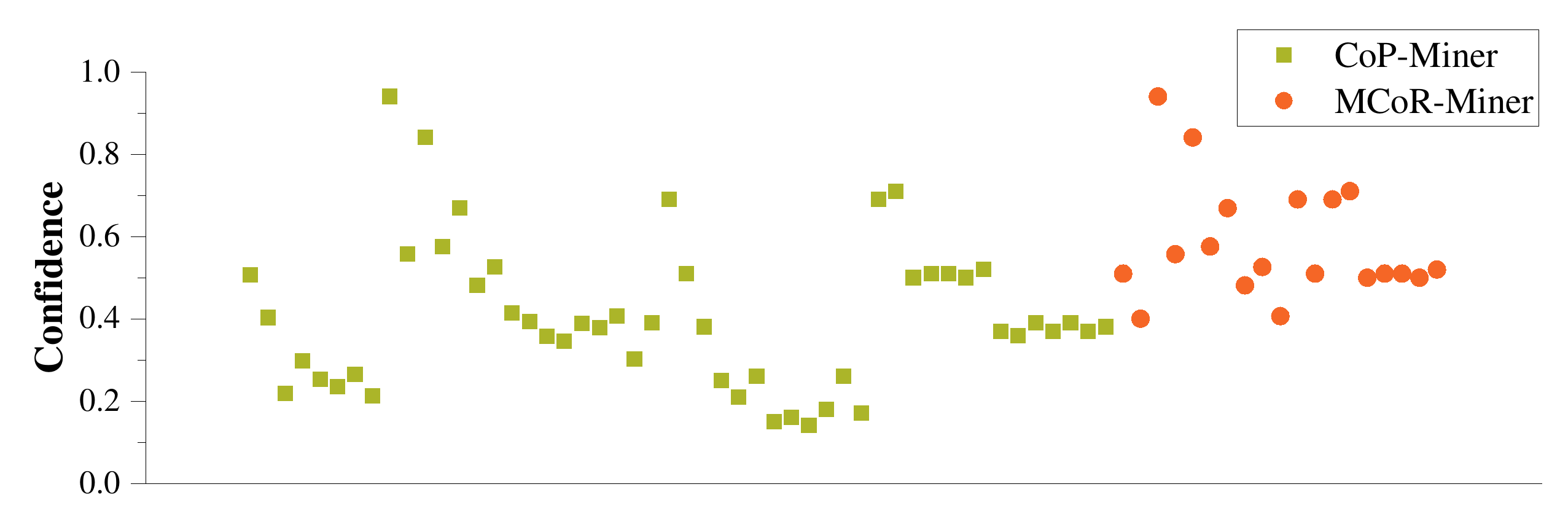}
    \caption{Comparison of confidence}
    \label{Comparison of confidence}
\end{figure}

\begin{figure}
    \centering
    \includegraphics[width=0.95\linewidth]{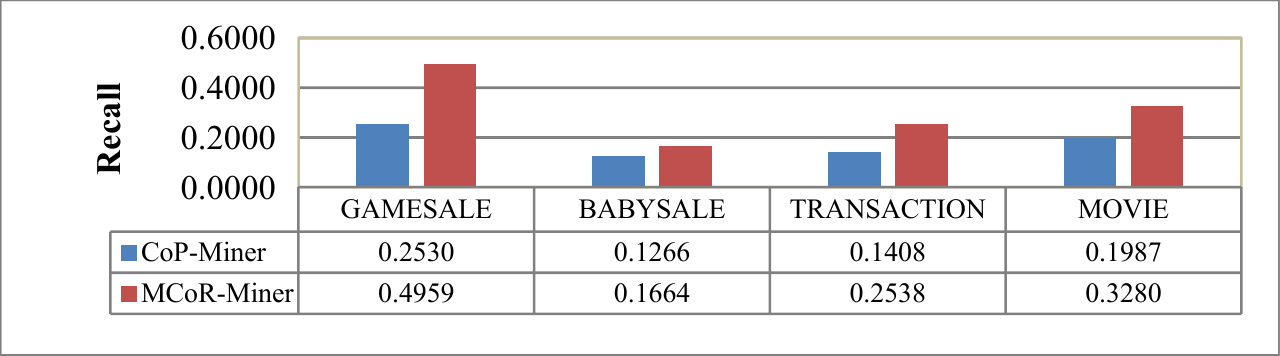}
    \caption{Comparison of recall}
    \label{Comparison of recall}
\end{figure}

\begin{figure}
    \centering
    \includegraphics[width=0.95\linewidth]{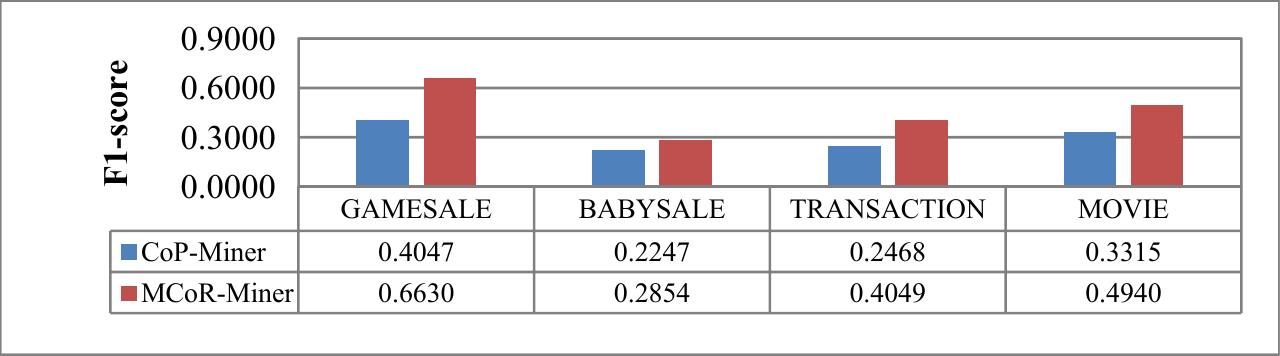}
    \caption{Comparison of F1-score}
      \label{Comparison of F1-score}
\end{figure}

The results indicate that our method has the following advantages.

\begin{enumerate} [1.]

\item MCoR-Miner can effectively prune the co-occurrence rules with low confidence. For example, Fig.\ref{Comparison of confidence} shows that the confidence for each maximal co-occurrence rule mined by MCoR-Miner is larger than 0.39, while the confidence for many of the co-occurrence patterns mined by CoP-Miner is less than 0.39, and the lowest is only 0.14. This is because MCoR-Miner discovers rules with high confidence and high frequency, while CoP-Miner only discovers rules with high frequency. Thus, MCoR-Miner can effectively prune co-occurrence patterns with low confidence, and the recommendation confidence is improved.


\item More importantly,  MCoR-Miner has better performance for recommendation task than CoP-Miner. For example, from Figs. \ref{Comparison of recall} and \ref{Comparison of F1-score}, on GAMESALE, the recall and F1-score of MCoR-Miner are 0.4959 and 0.6630, respectively, which are higher than those of CoP-Miner. This phenomenon can also be found on the other datasets, meaning that MCoR-Miner yields better recommendation performance than CoP-Miner. 

\item MCoR-Miner is a sequential rule-based recommendation method that is more interpretable than the learning-based recommendation method. The reason is as follows. MCoR-Miner discovers the rules with high confidence by calculating the support of patterns in the training set. Taking GAMESALE as an example, the supports of patterns $d$, $dd$, and $de$ are 1917, 970, and 763, respectively. Thus, MCoR-Miner generates two rules: \{$d$ $\to$ $d$, $d$ $\to$ $e$\}, whose confidences are 0.51 and 0.40, respectively. Therefore, MCoR-Miner has better interpretability. However, for a learning-based recommendation method, the recommendation model is very complex, and cannot be expressed intuitively by a mathematical formula.

\end{enumerate}

In summary, MCoR-Miner has better recommendation performance than CoP-Miner. More importantly, MCoR-Miner can form effective recommendation schemes, which can help decision-makers make correct decisions.


\section{CONCLUSION}
\label{section:CONCLUSION}

In nonoverlapping SPM or nonoverlapping sequential rule mining, if the prefix pattern or rule-antecedent is known, it is not necessary to discover all patterns or rules and then filter out irrelevant patterns or rules. This problem is called co-occurrence pattern mining or co-occurrence rule mining. To avoid discovering irrelevant patterns and to obtain better recommendation performance, and inspired by the concept of maximal pattern mining, we investigate MCoR mining. Unlike classical rule mining, which requires two parameters, minimum confidence and minimum support, MCoR mining does not need the minimum support, since it can be automatically calculated based on the support of the rule-antecedent and the minimum confidence. To effectively discover all MCoRs with the same rule-antecedent, we propose the MCoR-Miner algorithm. In addition to the support calculation, MCoR-Miner consists of three parts: preparation stage, candidate pattern generation, and screening strategy. To effectively calculate the support, we propose the DBI algorithm, which adopts depth-first search and backtracking strategies equipped with an indexing mechanism. MCoR-Miner uses the filtering strategy to prune the sequences without the rule-antecedent to avoid useless support calculation. It adopts the FET and BET strategies to reduce the number of candidate patterns, and employs the screening strategy to avoid finding the maximal rules by brute force. To evaluate the performance of MCoR-Miner, eleven competitive algorithms were implemented, and eight datasets were considered. Our experimental results showed that MCoR-Miner yielded better running performance and scalability than the other competitive algorithms. More importantly, compared with frequent co-occurrence pattern mining, MCoR-Miner had better recommendation performance.

In this paper, we focus on the nonoverlapping maximal co-occurrence  rule mining. The nonoverlapping sequential pattern mining is a kind of repetitive sequential pattern mining. We know that there are some similar mining methods, such as one-off sequential pattern mining and disjoint sequential pattern mining. Different mining methods will lead to different mining results. The one-off and disjoint maximal co-occurrence rule mining deserve further exploration. Moreover, for a specific  mining task, which mining method is the best is worth further study. More importantly, the sequence studied in this paper is a special sequence, which means that the sequence consists of items. However, general sequence consists of itemsets, each of which  has many ordered items. It is valuable to investigate frequent repetitive pattern mining and repetitive rule mining in general sequence database. 

\section*{Acknowledgement}
This work was partly supported by National Natural Science Foundation of China (61976240, 62120106008),  National Key Research and Development Program of China (2016YFB1000901), and Natural Science Foundation of Hebei Province, China (Nos. F2020202013).


%

 \vspace{-1cm}
\begin{IEEEbiography}[{\includegraphics[width=1in,height=1.25in,clip,keepaspectratio]{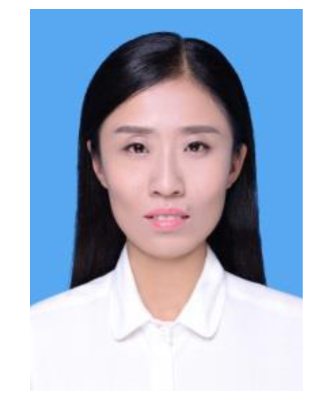}}]
{Yan Li}
received the PhD degree in Management Science and Engineering from Tianjin University, Tianjin, China. She is an associate professor with Hebei University of Technology. Her current research interests include data mining and supply chain management.
\vspace{-1cm}
\end{IEEEbiography}

\begin{IEEEbiography}[{\includegraphics[width=1in,height=1.25in,clip,keepaspectratio]{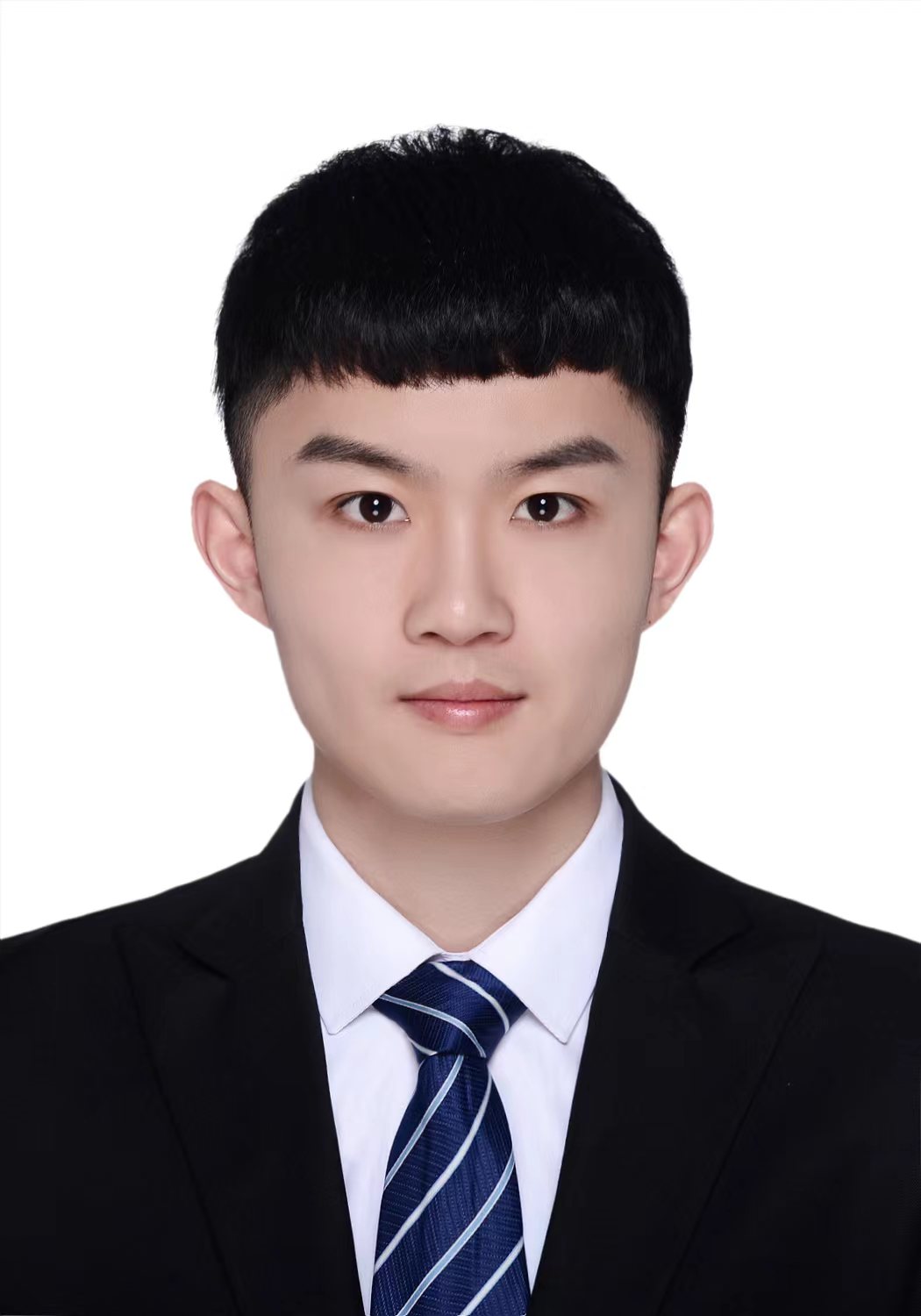}}]{Chang Zhang}  received the master degree in the School of Economics and Management from the Hebei University of Technology, Tianjin, China. His current research interests include data mining and machine learning.
\end{IEEEbiography}

\begin{IEEEbiography}[{\includegraphics[width=1in,height=1.25in,clip,keepaspectratio]{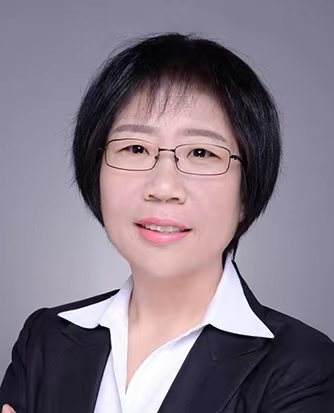}}]{Jie Li} received the Ph.D. degree in electrical engineering from the Hebei University of Technology, Tianjin, China, in 2002. She is currently a Professor with the School of Economics and Management, Hebei University of Technology. Her research interests include Big Data analytics in smart health care, business, and finance. 
\end{IEEEbiography}

\begin{IEEEbiography}
[{\includegraphics[width=1in,height=1.25in,clip,keepaspectratio]
{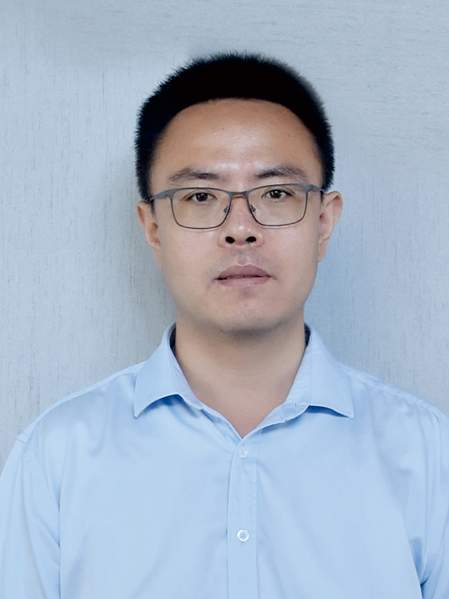}}]{Wei Song} received his Ph.D. degree in Computer Science from University of Science and Technology Beijing, Beijing, China, in 2008. He is a professor with the School of Information Science and Technology at North China University of Technology. His research interests are in the areas of data mining and knowledge discovery. He has published more than 40 research papers in refereed journals and international conferences.
\end{IEEEbiography}

\begin{IEEEbiography}[{\includegraphics[width=1in,height=1.25in,clip,keepaspectratio]{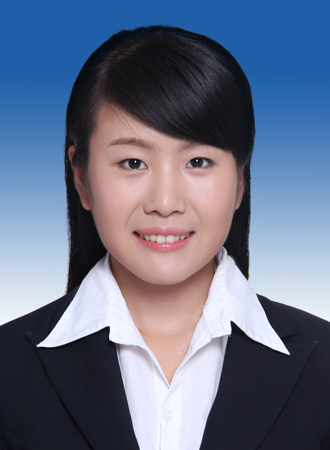}}]{Zhenlian Qi} received the Ph.D. in Environmental Science and Engineering, Harbin Institute of Technology, Harbin, China in 2020. She was a postdoc with Guangdong University of Technology, from 2020 to 2022. She is currently an Associate Professor with Guangdong Eco-Engineering Polytechnic, Guangzhou, China. Her research interests include environmental engineering, eco-big data, and data analytics.
\end{IEEEbiography}

\begin{IEEEbiography}[{\includegraphics[width=1in,height=1.25in,clip,keepaspectratio]{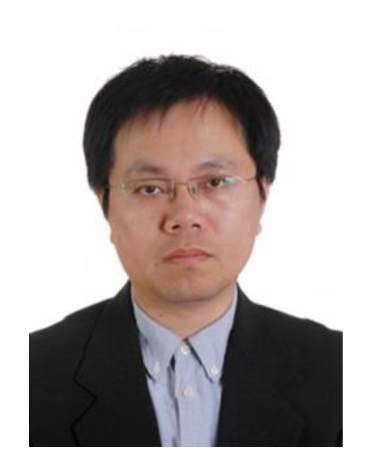}}]{Youxi Wu}received the Ph.D. degree in Theory and New Technology of Electrical Engineering from the Hebei University of Technology, Tianjin, China. He is currently a Ph.D. Supervisor and a Professor with the Hebei University of Technology. He has published more than 30 research papers in some journals, such as IEEE TKDE, IEEE TCYB, ACM TKDD, ACM TMIS, SCIS, INS, JCST, KBS, EWSA, JIS, Neurocomputing, and APIN. He is a senior member of CCF and a member of IEEE. His current research interests include data mining and machine learning.
\end{IEEEbiography}

\begin{IEEEbiography}[{\includegraphics[width=1in,height=1.25in,clip,keepaspectratio]{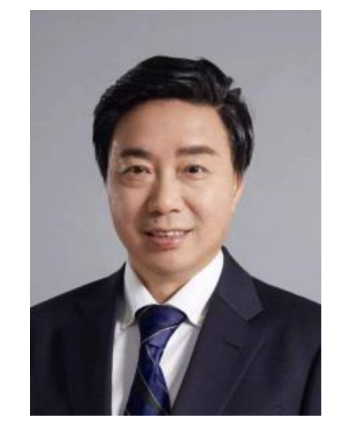}}]{Xindong Wu}
received the Ph.D. degree from the University of Edinburgh, Edinburgh, U.K. He is a Yangtze River Scholar with the Hefei University of Technology, Hefei, China. His current research interests include data mining, big data analytics, knowledge based systems, and Web information exploration. Dr. Wu was an Editor-in-Chief of  IEEE Transactions on Knowledge and Data Engineering (TKDE) (2005-2008) and is the Steering Committee Chair of the IEEE International Conference on Data Mining and the Editor-in-Chief of Knowledge and Information Systems. He is a Fellow of the American Association for the Advancement of Science and IEEE.
\end{IEEEbiography}

\end{document}